\documentclass{article}

\usepackage[preprint,nonatbib]{neurips_2022}

\usepackage[utf8]{inputenc} 
\usepackage[T1]{fontenc}    
\usepackage{hyperref}       
\usepackage{url}            
\usepackage{booktabs}       
\usepackage{amsfonts}       
\usepackage{nicefrac}       
\usepackage{microtype}      
\usepackage{xcolor}         

\newcommand{\mypara}[1]{\vspace{-2mm}\paragraph*{#1}}

\usepackage{graphicx}
\usepackage{float}
\usepackage{mathbbol}
\usepackage{amsthm}
\usepackage{amsmath}
\usepackage{amssymb}
\usepackage[all]{xy}
\usepackage{bbold}
\usepackage{xfrac}
\usepackage{stmaryrd}
\usepackage[capitalize]{cleveref}

\newcommand{\maths}[1]{{\mathbb #1}}  

\newcommand{\RR}{\maths{R}}
\newcommand{\NN}{\maths{N}}
\newcommand{\CC}{\maths{C}}

\newtheorem{Theorem}{Theorem}[section]
\newtheorem{Lemma}{Lemma}[section]    

\newtheorem{Def sec}{Definition}[section]   
\newtheorem{Definition}{Definition}[section]

\newcommand{\SE}{\mathrm{SE}}

\usepackage{soul}

\DeclareRobustCommand{\Adrien}[1]{ {\begingroup\sethlcolor{orange}\hl{(Adrien:) #1}\endgroup} }

\DeclareRobustCommand{\maks}[1]{ {\begingroup\sethlcolor{green}\hl{(Maks:) #1}\endgroup} }

\title{Equivalence Between SE(3) Equivariant Networks via Steerable Kernels and Group Convolution}

\author{%
  Adrien Poulenard \\
  Stanford University\\
  \texttt{adrien.poulenard@gmail.com} \\
  \And
  Maks Ovsjanikov \\
  LIX, Ecole Polytechnique, IP Paris\\
  \texttt{maks@lix.polytechnique.fr} \\
  \And
  Leonidas J.~Guibas \\
  Stanford University\\
  \texttt{guibas@cs.stanford.edu} \\
}

\begin{document}

\maketitle

\begin{abstract}
	A wide range of techniques have been proposed in recent years for designing neural networks for 3D data that are equivariant under rotation and translation of the input. Most approaches for equivariance under the Euclidean group $\mathrm{SE}(3)$ of rotations and translations fall within one of the two major categories. The first category consists of methods that use  $\mathrm{SE}(3)$-convolution which generalizes classical $\mathbb{R}^3$-convolution on signals over $\mathrm{SE}(3)$. Alternatively, it is possible to use \textit{steerable convolution} which achieves $\mathrm{SE}(3)$-equivariance by imposing constraints on $\mathbb{R}^3$-convolution of tensor fields. It is known by specialists in the field that the two approaches are equivalent, with steerable convolution being the Fourier transform of $\mathrm{SE}(3)$ convolution. Unfortunately, these results are not widely known and moreover the exact relations between deep learning architectures built upon these two approaches have not been precisely described in the literature on equivariant deep learning. In this work we provide an in-depth analysis of both methods and their equivalence and relate the two constructions to multiview convolutional networks.  Furthermore, we provide theoretical justifications of separability of $\mathrm{SE}(3)$ group convolution, which explain the applicability and success of some recent approaches. Finally, we express different methods using a single coherent formalism and provide explicit formulas that relate the kernels learned by different methods. In this way, our work helps to unify different previously-proposed techniques for achieving roto-translational equivariance, and helps to shed light on both the utility and precise differences between various alternatives. We also derive new TFN non-linearities from our equivalence principle and test them on practical benchmark datasets.

\end{abstract}

\section{Introduction}

The recent development of 3D deep learning has raised new challenges in manipulating and analyzing 3D data. The pose of 3D objects or the choice coordinate frame used to describe them is critical to the performance of many algorithms. Initially, algorithms were tested on datasets of aligned objects with consistent pose \cite{chang2015shapenet,wu20153d}. The observation of a generalization gap to arbitrary poses motivated the development of algorithms which are robust to change of position and orientation. Another motivation comes from physics. Physical laws are invariant to the choice of coordinates used to represent the physical phenomena; therefore, algorithms for learning physics should share the same behavior across different coordinate systems. This is formalized by the notion of equivariance. Intuitively a map is equivariant if its output changes in a predictable manner given a transformation of its input. In practice, this allows to obtain the same behavior under a transformation of the input, greatly improving the efficiency of learning. 


Given a group $G$ acting on two sets $A,B$, we say that a map $F:A \rightarrow B$ is equivariant to the action of $G$ ($G$-equivariant) if for all $g \in G$ acting on the input $x \in A$ of $F$ there is a corresponding action of $g$ on the output $F(x) \in B$ which is independent from $x$, i.e. for all $g \in G,x \in A$ we have $F(g.x) = g.F(x)$.

We are typically interested in the case where $F$ is a neural network with learnable parameters. A practical way to achieve equivariance is through convolution. For instance convolutional layers for image processing are equivariant to translation, this is also known as the weight sharing property. Image CNNs processes all regions of the image in the same way, and this property is often attributed as fundamental to their success. 


In addition to translation equivariance, as mentioned above, for 3D data, one is often interested in developing methods that would be equivariant under changes of orientation in 3D space. To address this challenge, in recent years many approaches have been developed for rotation equivariance of signals in 3D
\cite{thomas2018tensor, kondor2018n, cohen2018spherical, kondor2018clebsch, weiler20183d, worrall2018cubenet, esteves2019equivariant, anderson2019cormorant, fuchs2020se, deng2021vector, chen2021equivariant}.

Broadly, there exist two main classes of methods for designing rotationally equivariant neural networks for 3D data. The first option, 3D Steerable Convolution, which we review in \cref{sec:tfn}, is to operate directly over 3D signals in a convolution fashion relying on the concept of steerable kernel bases \cite{freeman1991design}. The idea is to achieve $\mathrm{SE}(3)$-equivariance (translation and rotation) by inheriting translation equivariance from convolution and rotation equivariance from the steerable basis. This has been investigated in \cite{weiler20183d} for voxel grids and in \cite{thomas2018tensor, kondor2018n} for point-clouds.

 Alternatively rather than manipulating signals over the 3D data it is also possible to consider functions defined over the group $G$ to which we want to be equivariant. In this context, a natural way to achieve $G$-equivariance, used in several recent works 
\cite{cohen2018spherical, esteves2019equivariant, chen2021equivariant} is through \textit{group convolution} which we review in \cref{sec:g_conv}. This operation generalizes convolution to arbitrary groups and is deeply related to equivariance. 


Remarkably, these two approaches are known to be equivalent by specialists of equivariant deep learning. Namely, steerable convolution can be seen as a Fourier transform of group convolution. For example, the equivalence in the $\mathrm{SE}(3)$ case is stated without proof in the slides by E. Bekkers \cite{Bekkers21slides} (Slide 94) and discussed briefly in the associated courses notes by \cite{Bekkers21} (Section 5.2). 
Nevertheless, despite the existence of such results, to the best of our knowledge 
there is a lack of dedicated treatment with explicit proofs regarding the exact relation between $\mathrm{SE}(3)$-equivariant architectures. Furthermore, these equivalence results do not seem to be currently known to the broader audience, and especially practitioners that have designed a wide range of alternative equivariant deep learning architectures for 3D point clouds. In particular recent works proposing $\mathrm{SE}(3)$-group convolution architecture do not mention this equivalence \cite{worrall2018cubenet, chen2021equivariant}. Our main goal, therefore, is to make these results both more explicit and more accessible to a broader audience interested in equivariant deep learning.


In this work we provide a detailed proof and analysis of the equivalence between steerable convolution and group convolution in the case of $\SE(3)$, as well as explicit formulas to translate expressions and filters between the two representations in \cref{th:tfn_vs_se3_conv_2}. We have not seen this translation formulas explicitly stated in equivariant networks literature. These results may not be surprising to specialists. Nevertheless we aim to establish and convey these theoretical observations to a broader audience, interested designing equivariant networks, as they show the relations and equivalence of different approaches.

Our discussion is general and applies to both the continuous setting of signals defined over $\mathbb{R}^3$ or $\mathrm{SE}(3)$ but also to the discrete setting when dealing with data represented as point clouds. Furthermore, our analysis allows to express different alternative approaches using a single coherent formalism and provide explicit formulas relating different methods. Finally, this theoretical equivalence allows us to define and interpret practical choices of non-linearities (\cref{sec:non_linearities}) within existing frameworks that we evaluate on real benchmark datasets in \cref{sec:results}.

\section{Background and Overview}
\label{sec:overview}

\mypara{Equivariance via multi-view networks}
Any non rotation equivariant network $F$ operating on signals over $f:\RR^3 \rightarrow \RR$ can be made rotation equivariant by considering the associated multi-view network defined by $\Tilde{F}(f)(x,R) := F(R.f)(x)$ where $R \in \mathrm{SO}(3)$ acts of $f$ by $(R.f)(x) := f(R^{-1}x)$. We immediately have the equivariance relation $\Tilde{F}(R.f)(x,H) = \Tilde{F}(f)(x,RH)$. Variants of this approach have been considered in \cite{qi2016volumetric,mehr2018manifold}. Such approaches however are limited as there is no information aggregation across multiple views at every layer, in practice pose information is aggregated at a later stage (at the last layer). In \cref{sec:multiview} we show that multi-view-CNNs are special cases of $\mathrm{SE}(3)$-CNNs.  

\mypara{Equivariance via steerable filters} Filters used by standard CNNs do not have predictable behavior under rotation. However it is possible do design filters which are equivariant to rotations \cite{freeman1991design}. Such filters are based on so-called steerable kernels bases which are vector valued maps $\kappa: \RR^d \rightarrow \RR^K$ ($d=2,3$) such that any rotation $R \in \mathrm{SO}(d)$ of the input $x \in \RR^d$ induces a rotation $D(R) \in \mathrm{SO}(d)$ of the output $\kappa(Rx) = D(R)\kappa(x)$. In 2D such bases can be obtained using the Fourier basis, while in 3D it can be obtained using Spherical Harmonics. This property of steerable kernels transfers to convolution features, for any scalar signal $f:\RR^d \rightarrow \RR$ the vector field $v(f) := f \ast \kappa$ satisfies $v(R.f)(x) = D(R)v(f)(R^{-1}x)$ where $R$ acts of $f$ by $(R.f)(x) := f(R^{-1}x)$. Steerable CNNs build on this idea to produce and operate on equivaiant fields satisfying the aforementioned equivariance relation. Several other works \cite{thomas2018tensor, kondor2018n} introduced equivalent constructions on point cloud domains, and we present the construction of \cite{thomas2018tensor} in \cref{sec:tfn}.

\mypara{Equivariance via group convolution} Another way to approach equivariance is to consider functions defined over the group $G$ to which we want to be equivariant. The action of $G$ on functions $f:G\rightarrow \RR$ is given by $(g.f)(x) := f(g^{-1}x)$. A natural way to achieve $G$-equivariance is through group convolution ($G$-convolution) which generalizes the usual notion of convolution to arbitrary groups \cite{kondor2018generalization}, group convolution takes a function and a kernel $f,\kappa:G \rightarrow \RR$ and produces a function $f \ast_G \kappa: G \rightarrow \RR$ satisfying the equivariance property $(g.f) \ast_G \kappa = g.(f \ast_G \kappa)$ for all $g \in G$. In \cite{kondor2018generalization}, Kondor and Trivedi show that \textit{any} equivariant continuous linear map between continuous functions on a compact group $G$ must be given by a group convolution. A similar result follows from \cite{cohen2018intertwiners, cohen2019general} when the group is unimodular. Group CNNs based on group convolutions inherit the same equivariance property. Multiple designs have been proposed, in 2D $\mathrm{SE}(2)$ \cite{cohen2016group}
and for various subgroups of $\mathrm{SE}(3)$: $\mathrm{SO}(3)$ \cite{cohen2018spherical}, $\RR^3 \rtimes C$ \cite{worrall2018cubenet}, $H$ \cite{esteves2019equivariant}, $\RR^3 \rtimes H$  \cite{chen2021equivariant} where $C,H \subset \mathrm{SO}(3)$ are the symmetry group of the cube and the icosahedral (symmetries of the icosahedron) groups respectively. We present the general definition of group convolution and CNNs in \cref{sec:g_conv}.

\mypara{Relations across different approaches (our contribution)} Thus we have reviewed three  approaches to $\mathrm{SE}(3)$-equivariance: the first is through multi-view augmentation a second is through 3D steerable CNNs relying on usual $\RR^3$-convolution of equivariant vector/tensor fields over $\RR^3$ with $\mathrm{SO}(3)$ steerable kernels and the other is through $\mathrm{SE}(3)$ convolution of scalar fields over $\mathrm{SE}(3)$.

Our main focus in this work is to provide a detailed proof and analysis of the equivalence between 3D steerable CNNs and $\mathrm{SE}(3)$ group CNNs stated in  \cite{Bekkers21} and to provide an explicit formulas to translate expressions and filters between the two representations (\cref{th:tfn_vs_se3_conv_2}) which has not been done in prior works.

The equivalence principle described in \cite{Bekkers21}  says that equivariant tensor field representations arise as harmonic decompositions of $\mathrm{SE}(3)$ signals and convolutions. An equivariant vector obtained via 3D steerable CNNs $v(R.f)(x) = D(R)v(f)(x)$ can be thought as a scalar function on $\mathrm{SE}(3) = \RR^3 \rtimes \mathrm{SO}(3)$,
$w: (x,R) \mapsto \sum_{ij} D_{ij}(R)v(f)_j(x) \in \RR$ and we show that performing 3D steerable convolution on $v$ is equivalent to a $\mathrm{SE}(3)$ convolution on $w$ (\cref{th:tfn_vs_se3_conv_1}).

As a part of our analysis we also analyze $\mathrm{SE}(3)$-convolution separability used in \cite{chen2021equivariant} to propose an efficient implementation of $\mathrm{SE}(3)$-CNN. Moreover, we provide additional insights showing that multi-view CNNs can be realized as $\mathrm{SE}(3)$ CNNs (\cref{th:multiview_cnn}).

While there are general results regarding equivalence of different approaches for $\mathrm{SE}(3)$ equivariance  \cite{Bekkers21, Bekkers21slides}, to the best of our knowledge, our detailed analysis has not been explicitly performed in existing literature \cite{kondor2018generalization, weiler20183d, cohen2018intertwiners, cohen2019general, lang2020wigner,weiler2021coordinate}. Furthermore the existing theories are stated in the continuous case while actual works like \cite{thomas2018tensor, kondor2018n, chen2021equivariant} operate on irregular domains like point clouds which are not invariant by rotation nor translation. In contrast our analysis covers both cases. Finally, we investigate $\mathrm{SE}(3)$-CNNs non-linearities in the harmonic domain and derive corresponding non linearities for Tensor Field Networks, \cite{thomas2018tensor} and similar constructions \cite{weiler20183d, kondor2018n} by equivalence. We show experimental results for \cite{thomas2018tensor} in \cref{sec:non_linearities}.

\section{Convolution and CNNs}

\label{sec:g_conv}
To avoid any integrability problems in our definitions and theorems, we will assume functions and kernels are measurable, bounded and compactly supported in the rest of this work. 

In this section we recall the general definition of group convolution (\cref{def:g_conv}), show how it specializes to practical cases and introduce notations for convolutional layers and networks (\cref{def:gcnn}). Convolution is the cornerstone of Convolutional Neural Networks for image analysis popularized by the famous AlexNet \cite{krizhevsky2012imagenet}. A key property of Euclidean convolution is translation equivariance, i.e. equivariance to the translation group $\RR^d$. Indeed, it is well-known that convolution is \textit{the only} linear operation that is translation-equivariant  \cite{hormander1960estimates}. This can be generalized to other groups by the notion of group convolution, and in particular the rotation groups $\mathrm{SO}(3)$ and roto-translation group $\mathrm{SE}(3)$. Group convolution offers a general  methodology to design neural networks which are equivariant to a group action, e.g., for $\mathrm{SO}(3)$ \cite{cohen2018spherical} or $\mathrm{SE}(3)$ \cite{chen2021equivariant}, among others. 
\begin{Definition}[Group convolution]
\label{def:g_conv}
For a group $G$, given an integrable function $f$ and kernel $\kappa$, where both $f, \kappa:G\rightarrow \RR$, convolution on $G$ for a measure $\nu$ is defined as:
    \[
    (f \ast_{\nu} \kappa)(x) := \int_G f(y) \kappa(y^{-1} \cdot x) \mathrm{d}\nu(y).
    \]
\end{Definition}
The classical definition of $G$-conv considers a uniform measure (also called Haar measure) which is invariant by the action of $G$ (see the survey of Esteves \cite{esteves2020theoretical} for a quick introduction and to Haar measures and \cite{hall2003lie, nachbin1976haar} for more details). We explicitly highlight the role of the measure $\nu$, as it becomes important in our discussion of discretization and specifically, point cloud convolution. We recall that the action of the group $G$ on a function $\kappa: G \rightarrow \RR$ is defined via the action of $G$ on its input: $(y.\kappa)(x) := \kappa(y^{-1} \cdot x)$. Similarly the action of $h \in G$ on a measure $\nu$ over $G$ is defined for any measurable subset $S \subseteq G$ by $(h.\nu)(S) := \nu(h^{-1} \cdot S)$. It is well-known (see \cref{sec:g_conv_eq_proof} in the supplementary material for a proof) that $G$-convolution is equivariant to these actions.

 
\begin{Theorem}[Equivariance of group convolution]
\label{th:gconv_equivariance}
Group convolution satisfies the following equiavriance property for any group element $h \in G$ and functions $f,g : G \rightarrow \RR$:
$
    (h.f) \ast_{h.\nu} g = h.(f \ast_{\nu} g)
$
\end{Theorem}
In other words, acting on the input of the convolution and the underlying measure by a group element $h$ is equivalent to acting on the output by the same element $h$. 
When $\nu$ is Haar the equivariance property simplifies and we can ignore the action on the measure as it is invariant. In \cite{kondor2018generalization} Kondor and Trivedi show that a linear map between functions on a compact group is equivariant iff it is convolutional. Cohen \textit{et al} show this results for uni-modular groups under reasonable assumptions \cite{cohen2018intertwiners, cohen2019general}.

\mypara{Euclidean convolution} Setting $G = \RR^d$ and $\nu$ to be the Lebesgue measure of $\RR^d$ we recover the classical notion of convolution $f \ast_{\RR^d,\nu} \kappa(x) = \int_{\RR^d} f(t)g(x-t) dt$.

\mypara{Point cloud convolution} 
Euclidean convolution can be adapted to the case of discrete point clouds where the input signal is only defined on the given set of points $X \subset \RR^3$. This can be formalized by setting $G = \RR^d$ and $\nu = \delta_X := \sum_i \delta_{X_i}$ the sum of Dirac measures at points of $X$. We then have $f \ast_{\RR^3, \delta_X} \kappa(X_i) = \sum_j f(X_j)\kappa(X_i - X_j)$. Note that a rotation of $\delta_X$ is equivalent to a rotation of $X$, for all $R \in \mathrm{SO}(3)$ we have $R.\delta_X = \delta_{R.X}$ (see
\cref{sec:dirac_measure_eq_proof} for a proof). This is a practical approach as the convolution only depends of values of $f$ and $\kappa$ on $X$ which can be stored in matrices. Multiple works investigate this idea \cite{atzmon2018point, hermosilla2018monte, hua2018pointwise, li2018pointcnn, xu2018spidercnn, thomas2019kpconv} (see also Sec. 3.3.2. in \cite{guo2020deep}). The main variations consist in developing methods to compensate biases induced by non-uniform or uneven sampling. Some works like \cite{atzmon2018point} normalize the signals, others use Monte Carlo integration with re-sampling \cite{hermosilla2018monte} and others learn the kernel functions themselves in addition to the weights \cite{thomas2019kpconv}.

\mypara{$\mathrm{SO}(3)$-convolution} $\mathrm{SO}(3)$-convolution is generally defined using the Haar measure of $\mathrm{SO}(3)$ (the volume measure of $\mathrm{SO}(3)$) which we denote by $\mu$ (see \cref{sec:so3_haar_measure} of the appendix of \cite{nachbin1976haar, chirikjian2000engineering} for more details). We have $f \ast_{\mathrm{SO}(3),\nu} \theta(R) := \int_{\mathrm{SO}(3)} f(H)\theta(H^{-1}R)d\mu(H)$.

\mypara{$\mathrm{SE}(3)$-convolution} Recall that the roto-translation group $\mathrm{SE}(3) = \RR^3 \rtimes \mathrm{SO}(3)$ is the semi-direct product of the translation group $\RR^3$ and the rotation group $\mathrm{SO}(3)$, and its composition law is given by: 
$(x, U).(y, V) = (Uy + x, UV), \ \ (t, H)^{-1} := (H^{-1}, -H^{-1}t)$. We set $\nu = \lambda \otimes \mu$ where $\lambda$ is either the Lebesgue measure of $\RR^3$ in the continuous case or $\lambda = \delta_X$ for a point cloud $X$, extending the notion of point cloud convolution to $\mathrm{SE}(3)$. We have: \newline $(f \ast_{\mathrm{SE}(3), \lambda \otimes \mu} g) (x, R) :=
  \int_{\RR^3} \int_{\mathrm{SO}(3)} f(t,H)g(H^{-1}(x-t), H^{-1}R) d\mu(H) d\lambda(t)$.

Note that in practice \textit{cross-correlation} is often preferred over convolution for Euclidean domains. Cross-correlation can be obtained from convolution by negating the input of the kernel, replacing $\kappa$ by $:x \mapsto \kappa(-x)$. In this work we refer to both as convolution as our analysis applies to both notions.

The notion of Group Convolutional Neural Network naturally extends the classical definition by replacing usual convolution by group convolution:
\begin{Definition}[Group convolutional neural network]
\label{def:gcnn}
Given a group $G$ endowed with a measure $\lambda$, a $G$-convolutional layer ($G$-conv) of kernels $(\kappa_j: G \rightarrow \RR)_j$, weights tensor $W$ and bias vector $b$ takes a family of functions $(f_i: G \rightarrow \RR)_i$ and a measure $\nu$ and outputs the family of functions on $G$ defined by: 
\[
\mathrm{Conv}_{G}(f,\nu, W, b)_i := \sum_{jk} W_{i,jk} f_j \ast_{\nu} \kappa_k + b_i
\]
A $G$-Convolutional Neural Network ($G$-CNN) consists in stacking $G$-conv layers. It takes a family of functions $f$ over $G$ and is defined recursively by:
\[
\mathrm{CNN}_G(f,\nu,W,b) := y^n, \ y^{k+1} := \xi(\mathrm{Conv}_{G}(y^{k}(f), \nu_k, W^{k+1}, b^{k+1})), \ y^0 := f
\]
where $\xi: \RR \rightarrow \RR$ is an activation function.
\end{Definition}

$G$-CNNs satisfy the same equivariance property as their layers, we provide a proof in \cref{sec:gcnn_eq_proof}: 
\begin{Theorem}[Equivariance of $G$-CNNs]
\label{th:gcnn_equivariance}
A $G$-CNN (or $G$-conv layer) satisfies the equivariance property:
$
\mathrm{CNN}_{G}(g.f,g.\nu,W, b) =
g.\mathrm{CNN}_{G}(f,\nu, W, b)
$
for all $g \in G$.
\end{Theorem}

\section{Our contributions}
Given the context above, our main goal is two-fold: firstly we aim to establish equivalence relations between existing works, and secondly, we aim to provide theoretical analysis for various notions of equivariance that have been proposed and used in recent works. In \cref{subsec:separability} we first analyze the separability of $\SE(3)$ group convolution and provide a theoretical justification for a recent practical approach, introduced in \cite{chen2021equivariant}. We then provide an interpretation of multiview convolution networks through the lens of $\SE(3)$ convolution in \cref{sec:multiview}. All proofs are given in the appendix which can be found in the supplementary material.

Our core result is the explicit equivalence relation between Tensor Field Networks (TFN) and $\SE(3)$ convolution, which we provide in \cref{sec:tfn_vs_se3_conv}, after briefly introducing TFN in \cref{sec:tfn}. Finally, with this analysis in hand, in \cref{sec:non_linearities} we propose a ReLu non-linearity in the harmonic domain, by formulating it for TFNs and interpreting it as ReLu non linearity for $\mathrm{SE}(3)$ conv, for which we provide practical evaluation in the following section.

\subsection{Separability of $\mathrm{SE}(3)$ convolution}
\label{subsec:separability}

In a recent work \cite{chen2021equivariant} Chen \textit{et al} consider a discretization of $\mathrm{SE}(3)$ convolution by using $\lambda = \delta_X$ and $\mu = \delta_H := \sum_{h \in H} \delta_{h}$ where $H \subset \mathrm{SO}(3)$ is the symmetry group of the icosahedron. Note that with this choice $\mathrm{SE}(3)$-convolution reformulates as a double sum over a sampling of $\mathrm{SE}(3)$ which is $6$ dimensional. Chen \textit{et al.} remark that $\mathrm{SE}(3)$ convolution can be computed by composing two 3D convolutions (a $\RR^3$ component and a $\mathrm{SO}(3)$ component) greatly reducing the cost. However, no theoretical justification for this approach was given. Below we provide such a justification and establish a result, which is useful in our further analysis below.

The key idea is that $\mathrm{SE}(3)$-convolution decomposes for so called separable functions. A function $f:\mathrm{SE}(3) = \RR^3 \rtimes \mathrm{SO}(3) \rightarrow \RR$ is called separable if it is a tensor product $f = \kappa \otimes \theta$ where $\kappa \otimes \theta(t,R) := \kappa(t)\theta(R)$. It is well known that general functions can be approximated by linear combination of separable functions in the L2 sense (see \cite{reed2012methods} on tensor product Hilbert spaces for more details). We formalize the two components of $\mathrm{SE}(3)$ convolution as follows:

\begin{Definition}[$\mathrm{SE}(3)$ convolution components]
\label{def:se3_conv_components}
The $\RR^3$ and $\mathrm{SO}(3)$ components of $\mathrm{SE}(3)$ convolution are defined for any function $f:\mathrm{SE}(3) \rightarrow \RR$ and kernels $\kappa:\RR^3 \rightarrow \RR$ and $\theta:\mathrm{SO}(3) \rightarrow \RR$ respectively by:
\begin{equation}
\begin{aligned}
f \ast_{\mu} (\delta_0 \otimes \theta)(x,R)
&
:= 
\int_{\mathrm{SO}(3)} \hspace{-3mm} f(x,H)\theta(H^{-1}R) d\mu(H)
\end{aligned}
\label{eq:sep_conv_1}
\end{equation}
\begin{equation}
\begin{aligned}
f \ast_{\lambda} (\kappa \otimes \delta_I)(x,R)
&
:= 
\int_{\RR^3} f(t,R)\kappa(R^{-1}(x-t)) d\lambda(t)
\end{aligned}
\label{eq:sep_conv_2}
\end{equation}
\end{Definition}

Separability of $\mathrm{SE}(3)$ convolution can be stated as follows:

\begin{Lemma}[Separable convolution factorization]
\label{lemma:separable_conv}
For any function $f:\mathrm{SE}(3) \rightarrow \RR$
and kernels $\kappa: \RR^3 \rightarrow \RR$ and $\theta: \mathrm{SO}(3) \rightarrow \RR$ with compact support we have:
$
(f \ast (\kappa \otimes \delta_I) ) \ast (\delta_0 \otimes \theta)
=
f \ast (\kappa_1 \otimes \kappa_2)(x, R).
$
\end{Lemma}

We can expand \cref{lemma:separable_conv} to convolution layers showing that $\mathrm{Conv}_{\mathrm{SE}(3)}$ layers can be computed as composition of $\RR^3$ and $\mathrm{SO}(3)$ separable components layers:
\begin{Theorem}[Separability $\mathrm{SE}(3)$-convolution layer]
\label{th:conv_approx}
Let $(\kappa_j:\RR^3 \rightarrow \RR)_j$ and $(\theta_k: \mathrm{SO}(3) \rightarrow \RR)_k$ two finite kernel bases. We define the $\RR^3 \times I$ and $0 \times \mathrm{SO}(3) $ convolution layers respectively for any finite family of functions $f_j: \RR^3 \rightarrow \RR$ and any weight tensors $A, B$ and bias vectors $a,b$ of compatible size  by:
$
\mathrm{Conv}_{\RR^3 \times I}(f, \lambda, A, a)_i := \sum_{jk} A_{ijk} f_k \ast_{\lambda} (\kappa_j \otimes \delta_I) + a_i,
\ \
\mathrm{Conv}_{0 \times \mathrm{SO}(3)}(f, \mu, B, b)_i := \sum_{jk} B_{ijk} f_k \ast_{\mu} (\delta_0 \otimes \theta_j)+b_i
$ the $\mathrm{SE}(3)$-convolution decompose as follows:
\[
\begin{aligned}
\mathrm{Conv}_{\mathrm{SE}(3)}(f, B.A, \lambda \otimes \mu ,b)_i
&
:=
\sum_{jkl} C_{ijkl} f_l \ast_{\lambda \otimes \mu} (\kappa_j \otimes \theta_k) + b_i
\\
&
=
\mathrm{Conv}_{0 \times \mathrm{SO}(3)}(\mathrm{Conv}_{\RR^3 \times I}(f, \lambda, A, 0), \mu, B, b) 
\end{aligned}
\]
where $B.A$ is the product tensor defined by $(B.A)_{ijkl} := \sum_m B_{ijm}A_{mkl}$. In particular a $\mathrm{SE}(3)$-convolution layer with weight tensor $C$ can be recovered by setting $A_{ni+j,kl} = C_{ijkl}$ and $B_{ij,na+b} = \delta_{ia}\delta_{jb}$ where $n$ is the dimension of $C$'s second axis.
\end{Theorem}
\subsection{Relation between multiview convolutional networks and $\mathrm{SE}(3)$-convolution.}
\label{sec:multiview}

While group convolution relies on both the signals and learned filters to be defined on the group $G$, a simpler approach is to maintain signals defined on $\RR^3$, and use \textit{rotation augmentation} to introduce a richer, rotation-aware network \cite{qi2016volumetric, mehr2018manifold}.



Specifically, given a $\RR^3$-CNN $\mathrm{CNN}_{\RR^3}$ its rotation augmented (multiview) version is given by:
\begin{equation}
    \label{eq:rot_augmented_cnn}
    \widetilde{\mathrm{CNN}}_{\RR^3}(f, \lambda, W)(x, R) := \mathrm{CNN}_{\RR^3}(R.f,R.\lambda, W)(x)
\end{equation}
It is relatively easy to see that rotation augmented $\RR^3$ CNNs are particular case of $\mathrm{SE}(3)$-CNNs where only the $\RR^3$ component is used (\cref{def:se3_conv_components}), this is consistent with the fact that multi-view networks do not share information across views by construction.
\begin{Theorem}[Multi-view CNNs as $\mathrm{SE}(3)$-CNNs]
\label{th:multiview_cnn}
For any $\RR^3$ CNN $\mathrm{CNN}_{\RR^3}$, denoting $\tilde{f}(x,R) := f(x)$ we have:
$
\widetilde{\mathrm{CNN}}_{\RR^3}(f, \lambda, W)(x,R) =
\mathrm{CNN}_{\RR^3 \times I}(\tilde{f}, \lambda, W)(R^{-1}x,R^{-1}).
$ where $\mathrm{CNN}_{\RR^3 \times I}$ is defined by replacing the $\mathrm{Conv}_{\RR^3}$ layers by $\mathrm{Conv}_{\RR^3 \times I}$ layers from \cref{th:conv_approx} in $\mathrm{CNN}_{\RR^3}$.
\end{Theorem}

\subsection{Overview of 3D steerable CNNs and Tensor Field Networks}
\label{sec:tfn}
In this section we recall the construction of 3D steerable CNNs \cite{weiler20183d} and Tensor Field Networks \cite{thomas2018tensor} although presented differently the two methods are quite similar, the difference is the input domain. 3D steerable CNNs \cite{weiler20183d} operate on voxel grids and are formulated in terms of continuous $\RR^3$ convolution (set $\lambda$ as the Lebesgue measure in the following description) while TFN operates on pointclouds and is based on pointcloud convolution (set $\lambda = \delta_X$ for a pointcloud $X$). We adopt the TFN terminology as we are more interested in pointcloud analysis. TFN can be viewed as an extension of pointcloud CNNs ($\lambda = \delta_X$) (or $\RR^3$ CNNs) to achieve $\mathrm{SE}(3)$-equivariance. The TFN design is a case of point-cloud convolution CNN as defined in \cref{sec:g_conv} while \cite{weiler20183d} is an equivalent formulation for 3D voxel based convolution. 

The key observation behind the TFN design is that Spherical Harmonics can be used to construct equivariant kernel bases called steerable bases and that the equivariance properties of such bases transfer to the associated convolution features. We recall that Spherical harmonics are homogeneous polynomial functions, over $\RR^3$, for each $\ell \in \NN^*$ there are $2\ell + 1$ degree $\ell$ spherical harmonics, a key property of spherical harmonics is their rotation equivariance, denoting by $Y_{\ell}:\RR^3 \rightarrow \RR^{2\ell+1}$ the vector of degree $\ell$ spherical harmonics Then for all  and $R \in \mathrm{SO}(3)$ there exist a matrix $D^{\ell}(R) \in \mathrm{SO}(2\ell+1)$ called the associated (type $\ell$) Wigner matrix, such that for all $x\in\RR^3$ we have
$
Y_{\ell}(Rx) = D^{\ell}(R)Y_{\ell}(x)
$
(see \cref{th:rotation_spherical_harmonics}). Spherical harmonics are used to build steerable kernel bases defined for all $x \in \RR^3$ by:
\begin{equation}
    \label{eq:steerable kernels}
    \kappa^{\ell}_{rm}(x) := \varphi_r(\Vert x \Vert_2) Y_{\ell m}(x)
\end{equation}
where $(\varphi_r: \RR_+ \rightarrow \RR)_r$ are radial functions. In practice in \cite{thomas2018tensor, weiler20183d} the kernels are fixed and the network learns coefficients in these kernel bases. A more recent work \cite{fuchs2020se} learns rotation invariant attention weights introducing a non linear deformation of the kernels but this goes beyond our analysis of linear convolution layers. We discuss possible choices of radial functions in \cref{sec:steerable_bases}. 
	
A key observation of \cite{thomas2018tensor, kondor2018n, weiler20183d}  is that given a function $f:\RR^3 \rightarrow \RR$ and a steerable basis $\kappa$ we have: 
$
(R.f) \ast_{\RR^3} \kappa_{r,:}^{\ell}(x)
=
D^{\ell}(R) f \ast_{\RR^3} \kappa_{r,:}^{\ell}(R^{-1}x)
$. This property can be formalized with the concept of equivariant features maps. A type $\ell$ equivariant feature map is a map $v^{\ell}:\RR^3 \rightarrow \RR^{2\ell+1}$ equipped with the $\mathrm{SO}(3)$ action, defined as $(R.v^{\ell})(x) := D^{\ell}(R)v^{\ell}(R^{-1}x)$. A TFN layer takes a collection of equivariant feature maps and transforms it linearly to another collection of equivariant features maps via standard $\RR^3$ convolution with a steerable basis. A second observation is that
\begin{equation}
\label{eq:composite_feature}
(R.v)^{\ell} \ast_{\RR^3} \kappa_{r,:}^{\ell'}(x) = D^{\ell}(R) \otimes D^{\ell'}(R) \left(v^{\ell} \ast_{\RR^3} \kappa_{r,:}^{\ell'}(R^{-1}x)\right).
\end{equation}
 where $\otimes$ is the standard tensor (Kronecker) product.
We see that convolution between equivariant features and a steerable basis introduces more complex equivariance properties. However it turns out that tensor products of Wigner matrices can be decomposed: For any $L, \ell, \ell' \in \NN$ with $|\ell - \ell'| \leqslant L \leqslant \ell + \ell'$ there exits a ``Clebsch-Gordan'' tensor $Q^{L,(\ell,\ell')} \in \RR^{(2L+1,2\ell+1,2\ell'+1)}$ (see \cref{sec:clebsch_gordan_coefficients} and \cite{chirikjian2000engineering} for more details) such that: 
$
D^L Q^{L,(\ell,\ell')} = Q^{L,(\ell,\ell')} D^{\ell} \otimes D^{\ell'}
$ (\cref{th:clebsh_gordan_decomposition}). The Clebsch-Gordan tensor $Q^{L,(\ell,\ell')}$ is given by the Clebsch Gordan coefficients and is used to decompose ``composite'' features from \cref{eq:composite_feature} into simpler features giving the general definition of TFN layer:
\begin{Definition}[TFN Layer]
A $\mathrm{TFN}$ layer of weight tensors $W$ and bias vector $b$ takes a collection $(f^{\ell}_{:,c}: \RR^3 \rightarrow \RR^{2\ell+1})_{\ell c}$ of equivariant vector features maps as input and outputs a new collection defined by:
\[
\mathrm{TFN}(f, W, \lambda, b)^L_{:,d} := \sum_{\ell,\ell',c,r} W^{(\ell, \ell'),L}_{d,cr} Q^{L,(\ell, \ell')}  f^{\ell}_{:,c} \ast_{\RR^3, \lambda} \kappa^{\ell'}_{r,:} + b^{L}
\]
where $b^L$ is non zero only for $L=0$ and the sum ranges over indices for which the main term is defined, that is $|\ell-\ell'| \leqslant L \leqslant \ell + \ell'$. 
\label{def:tfn_layer}
\end{Definition}
By construction a TFN layer satisfies the equivariance property $\mathrm{TFN}(R.f, R.\lambda, W, b)^L_{:,d}(x) = D^{\ell}(R)\mathrm{TFN}(f, \lambda, W, b)^L_{:,d}(R^{-1}x)$.

\subsection{Relation between 3D Steerable convolution / TFN and $\mathrm{SE}(3)$-convolution}
\label{sec:tfn_vs_se3_conv}
In this section we investigate the relation between TFN and $\mathrm{SE}(3)$-CNNs, and show that a TFN layer can be viewed as a harmonic decomposition (in the basis of Wigner coefficient functions) of an $\mathrm{SE}(3)$ convolution layer. Moreover, we establish a one-to-one correspondence between TFN and $\mathrm{SE}(3)$ CNN weights. We start by recalling some facts about harmonic analysis on $\mathrm{SO}(3)$, (we refer to \cref{sec:wigner_matrix} or \cite{chirikjian2000engineering} for more details). The Wigner matrices coefficients form a Hilbert basis of $L^{2}(\mathrm{SO}(3))$, in particular any square integrable function $f:\mathrm{SO}(3) \rightarrow \RR$ uniquely decomposes in the Wigner basis, i.e. for all $R \in \mathrm{SO}(3)$:
$
f(R) = \sum_{\ell \geqslant 0} \langle f^{\ell}, D^{\ell}(R)\rangle
$
where for each $\ell \in \NN$, $f^{\ell} \in \RR^{(2\ell+1, 2\ell+1)}$ is the matrix of coefficients of $f$ associated with $D^{\ell}$. Furthermore we have the equivariance property $(R.f)^{\ell} = D^{\ell}(R)f^{\ell}$ where $(R.f)(x) := f(R^{-1}x)$. 

Thus given a function $f:\mathrm{SE}(3) \rightarrow \RR$ we have a decomposition $f(x,R) = \sum_{\ell \geqslant 0} \langle f^{\ell}(x) , D^{\ell}(R)\rangle$ for all $(x, R) \in \mathrm{SE}(3)$. Moreover the coefficient matrices $f^{\ell}$, which can be thought of as functions $f^{\ell}:\RR^3 \rightarrow \RR^{(2\ell+1, 2\ell+1)}$, are equivariant feature maps as defined in \cref{sec:tfn} since $(R.f)^{\ell}(x) = D^{\ell}(R)f^{\ell}(R^{-1}x)$ where $(R.f)(x, H) := f(R^{-1}x, R^{-1}H)$. It therefore follows that a collection of $\mathrm{SE}(3)$ functions, $f_i: \mathrm{SE}(3) \rightarrow \RR$, can be seen a TFN layer input, with doubly indexed channel axis, by identifying each function $f_i$ with its feature maps  $(f_i^{\ell})_{\ell}$. 

Conversely we can interpret TFN features as coefficients of functions over $\mathrm{SE}(3)$ in the Wigner basis. A typical TFN layer takes a single indexed collection $(f^{\ell}_{:,c})_{\ell,c}$ of equivariant features. We can always complete such collection by zeros and reshape it into a doubly indexed collection of $2\ell+1$ by $2\ell+1$ matrices $(f^{\ell}_{:, :, c})_{\ell,c}$. The same operation can be done on the output and the weights can be adapted accordingly (extend with zeros and reshape) essentially obtaining a sparse representation of the original TFN layer. Now that we see how TFN features can be associated with $\mathrm{SE}(3)$ functions decomposed in the Wigner basis we describe the action $\mathrm{SE}(3)$ convolution layers on these decomposition. Essentially, we express the Wigner coefficients of $\mathrm{Conv}_{\mathrm{SE}(3)}(f, W)$ as a function of the Wigner coefficients of the input $f$ starting with the separable components from \cref{def:se3_conv_components}:
\begin{Lemma}[Wigner decomposition of separable $\mathrm{SE}(3)$ convolution]
\label{lemma:harmonic_sep_se3_conv}
Let $f: \mathrm{SE}(3) \rightarrow \RR$, the Wigner decomposition at $x \in \RR^3$ of the $\RR^3$ (resp. $\mathrm{SO}(3)$) component of $\mathrm{SE}(3)$ convolution (\cref{def:se3_conv_components}) between a function $f:\mathrm{SE}(3) \rightarrow \RR$ and steerable kernel basis over $\RR^3$ $(\kappa^{\ell'}_{rm'})_{r\ell'm'}$ (resp. a kernel $\theta:\mathrm{SO}(3) \rightarrow \RR$) is given by:
\[
\begin{aligned}
&
(f \ast_{\lambda} (\kappa^{\ell'}_{rm'} \otimes \delta_I))^L_{ij}(x)
=
\sum_{\ell}  Q^{L, (\ell, \ell')}_{i,:,:}  \left(f^{\ell} \ast_{\RR^3, \lambda} \kappa^{\ell'}_{r}\right)(x) Q_{:,m',j}^{ (\ell, \ell'), L}
\\
&
(f \ast_{\mu} (\delta_0 \otimes \theta))^{L}(x) =  f^{L}(x)\theta^{L}
\end{aligned}
\]
where the sum ranges over indices $\ell$ such that $|\ell - \ell'| \leqslant L \leqslant \ell + \ell'$.
\end{Lemma}
The Clebsch-Gordan coefficients appear in the expression of the $\RR^3$ component because it involves convolution between Wigner coefficient matrices $f^{\ell}$ steerable kernels $\kappa^{\ell'}$ that must be decomposed like in the TFN setting. The expression for the $\mathrm{SO}(3)$ component follows from the expression of $\mathrm{SO}(3)$ convolution in the Wigner basis which is used in \cite{cohen2018spherical}. Combining the results of \cref{lemma:harmonic_sep_se3_conv} and \cref{th:conv_approx} we obtain the expression of $\mathrm{Conv}_{\mathrm{SE}(3)}$ layers in the Wigner basis. We write the expression of the extended TFN layers described above side by side for easy visual comparison:
\begin{Theorem}[Wigner decomposition of $\mathrm{SE}(3)$ convolution and TFN layers]
 \label{th:tfn_vs_se3_conv_1}
 
For a kernel basis of the form $(\kappa_{rm'}^{\ell'} \otimes D^{L}_{Mn})_{r\ell'm',LMn}$ the $\mathrm{Conv}_{\mathrm{SE}(3)}$ and $\mathrm{TFN}$ layers have the following Wigner basis expansion for any collection of $\mathrm{SE}(3)$ functions $(f_c:\mathrm{SE}(3) \rightarrow \RR)_c$: 
\[
\begin{aligned}
\mathrm{Conv}_{\mathrm{SE}(3)}(f, W, b)^L_{ij,d}
&
:=
\sum_{\ell\ell'crm}
 Q^{L, (\ell, \ell')}_{i,:,:}  \left(f^{\ell}_{:,m,c} \ast_{\RR^3, \lambda} \kappa^{\ell'}_{r}\right) 
 \sum_{m'M} Q_{m,m',M}^{ (\ell, \ell'), L} W^{\ell', L}_{jd,crm'M} + b^L_d
 \\
 \mathrm{TFN}(f, V, b)^L_{ij,d}
 &
 := 
 \sum_{\ell\ell'crm}  Q^{L,(\ell, \ell')}_{i,:,:}  \left(f^{\ell}_{:,m,c} \ast_{\RR^3, \lambda} \kappa^{\ell'}_{r}\right) V^{(\ell, \ell'),L}_{jd,crm} + b^L_d
\end{aligned} 
\]
where $b^L_d := b_d$ iff $L = 0$ and $0$ otherwise.
\end{Theorem}
We immediately see from \cref{th:tfn_vs_se3_conv_1} that an $\mathrm{SE}(3)$-conv layer is a particular case of TFN layers as $\mathrm{SE}(3)$-conv weights linearly transform into TFN weights via the Clebsch Gordan tensors. Conversely we can show an equivalence between $\mathrm{SE}(3)$-conv and TFN layers as this transform is invertible, more precisely we have:
\begin{Theorem}[Equivalence between TFN and $\mathrm{SE}(3)$ convolution]
\label{th:tfn_vs_se3_conv_2}
We have a bijective linear map $\iota$ mapping $\mathrm{SE}(3)$-conv weights to $\mathrm{TFN}$ weights and its inverse mapping $\mathrm{TFN}$ weights $\mathrm{SE}(3)$-conv weights defined by:
\[
\iota(W)^{L,(\ell, \ell')}_{jd,crm}
:=
\sum_{m'M} Q_{m,m',M}^{ (\ell, \ell'), L} W^{\ell', L}_{jd,crm'M} \ \ \ \
\iota^{-1}(V)^{\ell',L}_{jd,crm'M}
:=
 \sum_{\ell,m}\frac{2\ell+1}{2L+1}Q^{L,(\ell,\ell')}_{M,m,m'} V^{L, (\ell, \ell')}_{jd,crm}
\]
for any collection of functions $(f_c:\mathrm{SE}(3)\rightarrow \RR)_c$ we have:
\[
\mathrm{Conv}_{\mathrm{SE}(3)}(f, \nu, W, b)
 = 
\mathrm{TFN}(f, \nu, \iota(W), b),
\ \ \ \ \
\mathrm{TFN}(f, \nu, V, b)
=
\mathrm{Conv}_{\mathrm{SE}(3)}(f,\nu, \iota^{-1}(V), b). 
\]
\end{Theorem}

A consequence of \cref{th:tfn_vs_se3_conv_1} and \cref{th:tfn_vs_se3_conv_2} is that TFN is a practical way to implement $\mathrm{SE}(3)$ CNNs. Discretizing $\mathrm{SE}(3)$-convolution is challenging especially the $\mathrm{SO}(3)$ component. A benefit of the TFN representation compared to other discretizations like \cite{chen2021equivariant} is full $\mathrm{SO}(3)$-equivariance (instead of a discrete subgroup). TFN is also a relatively efficient way to implement $\mathrm{SE}(3)$-convolution which is 6 dimensional, in practice however TFN layers are band-limited to limit the number of coefficients which limits its ability to model high frequency directional signals.

\subsection{$\mathrm{SE}(3)$ activations in the harmonic domain}
\label{sec:non_linearities}

An important aspect of designing equivariant networks is the exact choice of non-linearity. Our analysis above (\cref{th:tfn_vs_se3_conv_2}) shows the equivalence between TFNs which operate in the harmonic domain and SE(3) convolution, which is defined directly on the roto-translational group elements. Inspired by this analysis, we propose a simple pointwise non-linearity, which is based on navigating between these two representations: evaluating Wigner coefficients via inverse Wigner Transform (WT) $\mathcal{W}^{-1}(F) := \sum_{\ell} \langle F^{\ell}, D^{\ell}(R)\rangle$, applying a point-wise non linearity $\xi$ (ReLU) and converting back to Wigner coefficients with forward Wigner Transform $\mathcal{W}(f)^{\ell} := \int_{\mathrm{SO}(3)} f(R) D^{\ell}(R) d\mu(R)$ and its dicretized version $\mathcal{W}_G(f)^{\ell} := \frac{1}{|G|}\sum_{R\in G} f(R) D^{\ell}(R)$ for a finite subset $G \subset \mathrm{SO}(3)$. A very similar non linearity has been used for $\mathrm{SO}(3)$ convolution in $\cite{cohen2018spherical}$. The novelty here is to use it in the context of TFN while interpreting it as a ReLU activation of $\mathrm{SE}(3)$-CNNs. Also in the context of TFN \cite{poulenard2021functional} interprets TFN vector features as coefficients of function on the unit sphere $\mathcal{S}_2 \subset \RR^3$ and derives a similar non-linearity navigating between harmonic and pointwise domains on $\mathcal{S}_2$ with Spherical Harmonics Transforms (SHT). We refer to \cite{poulenard2021functional} for more discussion on TFN non linearities.

\begin{Theorem}[Equivariance of WT non linearities]
\label{th:activations_equivariance}
Denoting $\mathbf{Conv}_{\mathrm{SE}(3)}$ the Wigner decomposition of $\mathrm{Conv}_{\mathrm{SE}(3)}$ described in \cref{th:tfn_vs_se3_conv_1} the network defined by: 
\[
\mathbf{CNN}_{\mathrm{SE}(3)}(f, \nu, W, b) := y^n, \ y^{k+1} := \mathcal{F} \circ \xi \circ \mathcal{F}^{-1}(\mathbf{Conv}_{\mathrm{SE}(3)}(y^{k}, \nu_k, W^{k+1}, b^{k+1})), \ y^0 := f
\]
is a standard $\mathrm{SE}(3)$-CNN if $\mathcal{F} = \mathcal{W}$ is the Wigner transform. If $\mathcal{F} = \mathcal{W}_G$ where $G \subset \mathrm{SO}(3)$ is a finite subgroup then $F$ is equivariant to $G$, i.e. for all $(t,R) \in \RR^3 \times G$ and all $x \in \RR^3$ we have: 
\[
\mathbf{CNN}_{\mathrm{SE}(3)}((t,R).f, (t,R).\nu, W, b)^{L}(x) = D^{L}(R)\mathbf{CNN}_{\mathrm{SE}(3)}(f, \nu, W, b)^{L}(R^{-1}x-t).
\]
\end{Theorem}

\section{Experimental validation}
\label{sec:results}
Using our equivalence principle from \cref{th:tfn_vs_se3_conv_2} we can implement the Wigner decomposition of $\mathrm{SE}(3)$-CNNs with ReLU non linearities as described in \cref{sec:non_linearities} within a TFN pipeline. We adapt the TFN pipeline from \cite{poulenard2021functional}, by reshaping the tensors to match the $\mathrm{SE}(3)$ convolution format as described in \cref{sec:tfn_vs_se3_conv} while approximately keeping the same tensor size as the original architecture. We evaluate our modified TFN model on the ModelNet40 3D objects classification benchmark \cite{wu20153d} and report results in \cref{table:modelnet40}.  We consider two variants of our architecture, one with 60 $\mathrm{SO}(3)$ samples where $G$ is the icosahedral group (ours (ico)) and an other where $G$ consists 256 samples of $\mathrm{SO}(3)$ obtained with farthest point sampling to evaluate the impact of the density of the sampling. We compare our model to its discrete $\mathrm{SE}(3)$-CNN counterpart introduced in \cite{chen2021equivariant} and to other TFN non linearities implemented within the same TFN model retraining models from \cite{poulenard2021functional}. Results are reported in \cref{table:modelnet40}. We refer to \cite{poulenard2021functional} for more discussion about TFN non linearities and additional baselines for this benchmark.  We see in \cref{table:modelnet40} that our results are consistent with \cite{chen2021equivariant} and comparable to other TFN non linearities, although SHT non linearity from \cite{poulenard2021functional} seems to perform better. A possible reason might be that SHT can be approximated with fewer samples compared to WT as the unit sphere $\mathcal{S}_2$ is 2 dimensional while $\mathrm{SO}(3)$ is 3 dimensional. We also notice the impact of sampling density for computing WT, lower number of samples degrades equivariance. This has an important effect on generalization when training on aligned data and testing on rotated data as can be seen with the much wider generalization gap for ours(ico) (60 samples) compared to ours(256) (256 samples).

\begin{table}[ht]
\caption{Test classification accuracy on ModelNet40 \cite{wu20153d} in two train / test scenarios. $z$: aligned data augmented by random rotations around the z axis,  $\mathrm{SO}(3)$: augmentation by random $\mathrm{SO}(3)$ rotations. In Ours (ico) we set the $G \subset \mathrm{SO}(3)$ to be the 60 elements of the icosahedral group. In ours (256) we select 256 $\mathrm{SO}(3)$ samples via farthest point sampling. For TFN we use the original non linearities of \cite{thomas2018tensor}, TFN[gated] uses gated non-linearities from \cite{weiler20183d}, TFN[ReLU] uses SHT
with ReLU activation with 64 $\mathcal{S}_2$ samples. All networks are trained for 150 epochs.
\label{table:modelnet40}}
\begin{center}
\scalebox{0.84}{
\begin{tabular}{|c||c|c|c|c|c|c|}
\hline
& ours (ico) & ours (256)& $\mathrm{SE}(3)$-CNN \cite{chen2021equivariant} & TFN \cite{poulenard2021functional}& TFN[gated] \cite{poulenard2021functional} & TFN[ReLU] \cite{poulenard2021functional} \\
\hline \hline
$z / \mathrm{SO}(3)$ & 67.2 & 81.1 & - & 87.1 & 87.3 & 88.2
\\
$\mathrm{SO}(3) / \mathrm{SO}(3)$ & 88.2 & 88.4 & 88.3 & 87.5 & 88.5 & 89.3
\\
\hline
\end{tabular}}
\end{center}
\end{table}

\vspace{-3mm}\section{Conclusion}

In this work we have shown an equivalence between two seemingly different approaches to $\mathrm{SE}(3)$ equivariance. Namely we have shown that 3D steerable convolution is the Fourier (or Wigner) transform of $\mathrm{SE}(3)$-convolution. We also observed that multiview CNNs are a particular case of $\mathrm{SE}(3)$-CCNs.
Together with other results showing that equivariant linear layers must be convolutional, our analysis suggests that different existing designs, while equivalent in the infinite resolution, achieve trade-offs primarily due difference in sampling, being band-limited (for TFN and Steerable CNNs) and limited by finite sampling of rotations (for $\SE(3)$-conv). Although some additional analysis might be required to draw a definitive conclusion on whether equivariant linear layers are always convolutional in discrete settings like point cloud data. Our observations also suggest to explore non-linear equivariant layer designs (e.g. polynomial) as a path forward for research on equivariant networks contrary to the linear setting. We believe this path is still not extensively explored. 
\bibliographystyle{unsrt}
\bibliography{neurips_2022}



\appendix

\section{Background material on $\mathrm{SO}(3)$ harmonic analysis}

In this section we recall the definitions of harmonic bases we used in this work (Spherical Harmonics, Wigner Matrices, Clebsch-Gordan coefficients) as well as some of their fundamental properties. A general reference for this section is \cite{chirikjian2000engineering}.

\label{sec:harmonics}

\subsection{Spherical harmonics}
\label{sec:spherical_harmonics}
\begin{Definition}[Real Spherical harmonics]
Spherical harmonics are homogeneous polynomial functions, over $\RR^3$, for each $\ell \in \NN^*$ there are $2\ell + 1$ degree $\ell$ spherical harmonics $Y_{-m}^{\ell}, \dots, Y_{m}^{\ell}$ defined for any $X = (x,y,z) \in \RR^3$ by:
\[
\begin{aligned}
&
Y_{\ell,-m}(x,y,z) := \sqrt{\frac{2\ell+1}{2\pi}} \overline{\Pi}^{\ell}_m(x, y, z)B_m(x,y), \ \ m > 0
\\
&
Y_{\ell,0}(x,y,z):= \sqrt{\frac{2\ell+1}{4\pi}} \overline{\Pi}^{\ell}_0(x, y, z)
\\
&
Y_{\ell,m}(x,y,z) := \sqrt{\frac{2\ell+1}{2\pi}} \overline{\Pi}^{\ell}_m(x, y, z)A_m(x,y), \ \ \ m > 0
\end{aligned}
\]
where:
\[
\begin{aligned}
&
A_m(x,y) := \sum_{p=0}^m x^m y^{m-p} \cos((m-p)\frac{\pi}{2}),
\ \ \ 
B_m(x,y) := \sum_{p=0}^m x^m y^{m-p} \sin((m-p)\frac{\pi}{2})
\\
&
\overline{\Pi}^{\ell}_m(x, y, z) := \sqrt{\frac{(\ell-m)!}{(\ell+m)!}} \sum_{k=0}^{\lfloor (\ell - m)/2 \rfloor} (-1)^k 2^{\ell} \begin{pmatrix} \ell \\ k \end{pmatrix}\begin{pmatrix} 2\ell - 2k \\ \ell \end{pmatrix} \frac{(\ell - 2k)!}{(\ell - 2k - m)!} \Vert X \Vert^{2k}_2 z^{\ell - 2k - m}
\end{aligned}
\]
\end{Definition}

\begin{Theorem}[Hilbert basis of $L^2(\mathcal{S}_2)$]
The restriction of spherical harmonics to the unit sphere $\mathcal{S}_2 \subset \RR^3$ form a Hilbert basis of the space $L^2(\mathrm{S}_2)$ of square integrable functions on $\mathcal{S}_2$ for the surface measure of $\mathcal{S}_2$ whose density in the spherical coordinates $(\theta, \varphi) \in [0, 2\pi[ \times [0, \pi[$ is $\sin(\theta)$, that is, for any $\ell,m,\ell',m'$ we have:
\[
\int_{\varphi=0}^{2\pi} \int_{\theta=0}^{\pi} (Y_{\ell m}Y_{\ell' m'})(\cos(\varphi)\sin(\theta), \sin(\varphi)\sin(\theta), \cos(\theta)) \sin(\theta) \mathrm{d}\theta \mathrm{d}\varphi = \delta_{\ell,\ell'}\delta_{m,m'}.
\]
\end{Theorem}

We have described the real spherical harmonics. Complex spherical harmonics are related to real spherical harmonics by transition matrices. The transition matrix $C^{\ell} \in \mathcal{M}_{2\ell+1}(\CC)$ from complex to real spherical degree $\ell$ spherical harmonics is given by:
\[
\begin{cases}
\displaystyle C^{\ell}_{m,m} = \frac{-i}{\sqrt{2}} & \text{if}\ m<0\\
\displaystyle C^{\ell}_{m,m} = \frac{(-1)^m}{\sqrt{2}} & \text{if}\ m>0\\
\displaystyle C^{\ell}_{m,-m} = \frac{1}{\sqrt{2}} & \text{if}\ m<0\\
\displaystyle C^{\ell}_{m,-m} = \frac{i(-1)^m}{\sqrt{2}} & \text{if}\ m>0\\
\displaystyle C^{\ell}_{0,0} = 1 & \\
\displaystyle C^{\ell}_{j,k} = 0 & \text{otherwise}
\end{cases}
\]
the transition from real to complex harmonics is given by the transpose conjugate matrix $(C^{\ell})^*$. 

\subsection{Haar measure of $\mathrm{SO}(3)$}
\label{sec:so3_haar_measure}
The Haar measure of $\mathrm{SO}(3)$ is usually described  by its density with respect to Lebesgue's measure in the $z-y-z$ Euler angle parametrization which we now describe. For any $\theta \in \RR$ $R_X(\theta), R_Y(\theta), R_Z(\theta) \in \mathrm{SO}(3)$ the rotations of angle $\theta$ around the $x, y$ and $z$ axes respectively.
\begin{Theorem}[Euler angle parametrization of $\mathrm{SO}(3)$]
The Euler angles map 
\[
\begin{array}{rcl}
R_{ZYZ} : [0, 2\pi] \times [0, \pi] \times [0, 2\pi] &\longrightarrow \mathrm{SO}(3) \\
     (\alpha, \beta, \gamma) &\longmapsto R_{Z}(\alpha)R_{Y}(\beta)R_{Z}(\gamma) 
\end{array}
\]
is a surjection whose restriction to $[0, 2\pi[ \times [0, \pi[ \times [0, 2\pi[$ is injective and its image is a dense open subset of $\mathrm{SO}(3)$.
\end{Theorem}

\begin{Theorem}[Haar Measure of $\mathrm{SO}(3)$]
The riemannian metric of $\mathrm{SO}(3)$ is the the metric induced by the inclusion $\mathrm{SO}(3) \hookrightarrow \mathrm{M}_{3,3}(\RR) \simeq \RR^9$. The Haar measure $\mu$ of $\mathrm{SO}(3)$ is the measure induced by its riemannian metric. The pull back of the volume form $\omega$ associated with the riemannian metric of $\mathrm{SO}(3)$ by the Euler angles map $R$ has the following expression in the canonical coordinate system $(\alpha, \beta, \gamma)$ of $[0, 2\pi[ \times [0, \pi[ \times [0, 2\pi[$:
\[
R_{ZYZ}^*\omega = \frac{1}{8\pi^2}\sin(\beta) d\alpha \wedge d\beta \wedge d\gamma.
\]
that is, parametrizing $\mathrm{SO}(3)$ by $[0, 2\pi[ \times [0, \pi[ \times [0, 2\pi[$ the density of $\mu$ w.r.t. the Lebesgues measure $\lambda$ is given by $\frac{d\mu}{d\lambda}(\alpha, \beta, \gamma) = \frac{1}{8\pi^2}\sin(\beta)$.
\end{Theorem}

\subsubsection{The Wigner matrix}

\label{sec:wigner_matrix}
In this section we recall the definition of Wigner matrices and some of their properties and relations to spherical harmonics.

\begin{Definition}[Wigner Matrix]
Let $\ell \in \NN$, The Wigner $d^{\ell}$ matrix is defined for any angle $\beta$ by:
\[
\begin{array}{lcl}
d^{\ell}_{m'm}(\beta) &=& [(\ell+m')!(\ell-m')!(\ell+m)!(\ell-m)!]^{1/2}
\sum\limits_s \left[\frac{(-1)^{m'-m+s}}{(\ell+m-s)!s!(m'-m+s)!(\ell-m'-s)!} \right.\\
&&\left. \cdot \left(\cos\frac{\beta}{2}\right)^{2\ell+m-m'-2s}\left(\sin\frac{\beta}{2}\right)^{m'-m+2s} \right].
\end{array}
\]
for $m, m' \in [|-\ell, \ell|]$ where the sum over $s$ is over such values that the factorials are non-negative. The complex Wigner $\hat{D}^{\ell}$ is defined for any angles $\alpha, \beta, \gamma$ by:
\[
\hat{D}^{\ell}(\alpha, \beta, \gamma) := e^{-im'\alpha}d^{\ell}_{m'm}(\beta)e^{-im\gamma}, \ \ m, m' \in [|-\ell, \ell|].
\]
The real Wigner matrix $D^{\ell}$ is defined by:
\[
D^{\ell} := C^{\ell} \hat{D}^{\ell} (C^{\ell})^*
\]
\end{Definition}

\begin{Theorem}[$\mathrm{SO}(3)$-representations]
\label{th:so3_representations}
The Wigner matrix $D^{\ell}$ is a rotation matrix, to a map (also denoted $D^{\ell}$) $D^{\ell}: \mathrm{SO}(3) \rightarrow \mathrm{SO}(2\ell+1)$
\[
  \xymatrix{
    [0,2\pi] \times [0, \pi] \times [0,2\pi] \ar[rr]^{R_{ZYZ}} \ar[rd]_{D^{\ell}} & & \mathrm{SO}(3) \ar[ld]^{D^{\ell}}\\
    & \mathrm{SO}(2\ell+1) & 
  }
\]
The Winger matrix $D^{\ell}: \mathrm{SO}(3) \rightarrow \mathrm{SO}(2\ell+1)$ is a unitary $\mathrm{SO}(3)$ representation, i.e. for all $R,R'\in \mathrm{SO}(3)$ we have:
\[
D^{\ell}(RR') = D^{\ell}(R)D^{\ell}(R'), \ \ D^{\ell}(I) = I.
\]
The irreducible representation of $\mathrm{SO}(3)$ are given by the Wigner matrices $(D^{\ell})_{\ell \in \NN}$.
\end{Theorem}

\begin{Theorem}[Rotation of spherical harmonics]
\label{th:rotation_spherical_harmonics}
Denoting by $Y_{\ell}: \RR^3 \rightarrow \RR^{2\ell+1}$ the vector of degree $\ell$ real spherical harmonics, for any rotation $R \in \mathrm{SO}(3)$ and $x \in \RR^3$ we have:
\[
D^{\ell}(R)Y_{\ell}(x) = Y_{\ell}(Rx).
\]
\end{Theorem}

\begin{Theorem}[Hilbert basis of $L^2(\mathrm{SO}(3))$]
The Wigner matrix coefficients form a Hilbert basis of the the Hilbert space of square integrable functions on $\mathrm{SO}(3)$ for its Haar measure, more precisely: 
\begin{enumerate}
    \item For all $\ell,m,n$ and $\ell',m',n'$ such that the coefficients $D^{\ell}_{mn}$ and $D^{\ell'}_{m'n'}$ are defined we have:
\[
\int_{0}^{2\pi} \hspace{-3mm}\int_{0}^{\pi}
\hspace{-2mm}\int_{0}^{2\pi}
\hspace{-3mm}
(D_{mn}^{\ell}D_{m'n'}^{\ell'})(\alpha, \beta, \gamma) \sin(\beta)\mathrm{d}\alpha \mathrm{d}\beta \mathrm{d}\gamma 
=
\frac{8\pi^2}{2\ell+1}\delta_{\ell\ell'}\delta_{mm'}\delta_{nn'}.
\]
\item For any square integrable function $f:\mathrm{SO}(3) \rightarrow \RR$ uniquely decomposes in the Wigner basis, i.e. for all $R \in \mathrm{SO}(3)$:
\[
f(R) = \sum_{\ell \geqslant 0} \langle f^{\ell}, D^{\ell}(R)\rangle
\]
where for each $\ell \in \NN$, $f^{\ell} \in \RR^{(2\ell+1, 2\ell+1)}$ is the matrix of coefficients of $f$ associated with $D^{\ell}$.
\item Furthermore the coefficient matrices of $f$ satisfy the equivariance relation $(R.f)^{\ell} = D^{\ell}(R)f^{\ell}$ for all $\ell \in \NN, R\in \mathrm{SO}(3)$, where $(R.f)(x) := f(R^{-1}x)$. 

\end{enumerate}
\begin{proof}
We refer to \cite{chirikjian2000engineering} for 1) and 2). the proof of 3) is straightforward, by 2) and \cref{th:so3_representations} we have:
\[
\begin{aligned}
\sum_{\ell} \langle (R.f)^{\ell}, D^{\ell}(H) \rangle
=
(R.f)(H)
&
=
\sum_{\ell} \langle f^{\ell}, D^{\ell}(R^{-1}H) \rangle
\\
&
=
\sum_{\ell} \langle f^{\ell}, D^{\ell}(R)^{\top}D^{\ell}(H) \rangle
\\
&
=
\sum_{\ell} \langle D^{\ell}(R)f^{\ell}, D^{\ell}(H)\rangle
\end{aligned}
\]
by uniqueness of the Wigner decomposition 2) we have $(R.f)^{\ell} = D^{\ell}(R)$ for all $\ell$.
\end{proof}

\label{th:so3_hilbert_basis}
\end{Theorem}

\subsubsection{Clebsch Gordan Coeffcients}

\label{sec:clebsch_gordan_coefficients}

\begin{Definition}[Clebsch-Gordan coefficients]
\label{def:clebsch_gordan_coefficients}
The Clebsch-Gordan coefficients are defined for any $\ell,\ell' \in \NN$, $m \in \llbracket -\ell, \ell \rrbracket, m' \in \llbracket -\ell, \ell \rrbracket$ and $L \in \llbracket |\ell-\ell'|, \ell+\ell' \rrbracket$ such that $\ell \geqslant \ell'$ and $M \in \NN$ by:
\[
\begin{aligned}
&
\langle \ell,\ell';m,m'|\ell,\ell';L,M\rangle 
\\
:= \ 
&\delta_{M,m+m'} \sqrt{\frac{(2L+1)(L+\ell-\ell')!(L-\ell+\ell')!(\ell+\ell'-L)!}{(\ell+\ell'+L+1)!}}\ \times \\
&\sqrt{(L+M)!(L-M)!(\ell-m)!(\ell+m)!(\ell'-m')!(\ell'+m')!}\ \times \\
&\sum_k \frac{(-1)^k}{k!(\ell+\ell'-L-k)!(\ell-m-k)!(\ell'+m'-k)!(L-\ell'+m+k)!(L-\ell-m'+k)!}.
\end{aligned}
\]
The definition extends to $M < 0$ and $\ell < \ell'$ using the symmetry relations:
\[
\begin{aligned}
&
\langle \ell,\ell';m,m'|\ell,\ell';L,M\rangle := (-1)^{L-\ell-\ell'}\langle \ell,\ell';-m,-m'|\ell,\ell';L,-M\rangle
\\
&
\langle \ell,\ell';m,m'|\ell,\ell';L,M\rangle := (-1)^{L-\ell-\ell'} \langle \ell',\ell;m',m|\ell',\ell;L,M\rangle
\end{aligned}
\]
For convenience we define the sparse tensors $Q^{L, (\ell,\ell')} \in \RR^{(2L+1)\times (2\ell+1) \times (2\ell'+1)}$ and its transpose \newline $Q^{(\ell,\ell'),L} \in \RR^{(2\ell+1) \times (2\ell'+1) \times (2L+1)}$ by:
\[
Q^{(\ell,\ell'), L} := (C^{\ell} \otimes C^{\ell'}) \hat{Q}^{(\ell,\ell'), L} (C^{L})^*, \ \
Q^{L, (\ell,\ell')} :=  C^{L} \hat{Q}^{L, (\ell,\ell')}  ((C^{\ell})^* \otimes (C^{\ell'})^*)
\]
Where the tensors $\hat{Q}^{L, (\ell,\ell')}$ and $\hat{Q}^{L, (\ell,\ell'), L}$ are sparse tensors whose non zero coefficients are given by the Clebsch-Gordan coefficients:
\[
\hat{Q}^{(\ell,\ell'), L}_{m, m', m+m'} := \hat{Q}^{L, (\ell,\ell')}_{m+m', m, m'} := \langle \ell,\ell'; m,m' | \ell,\ell'; L, \left(m + m'\right) \rangle.
\]

\end{Definition}

\begin{Theorem}[Orthogonality relation of Clebsch-Gordan coefficients]
\label{th:clebsch_gordan_orthogonality}
For any $\ell, \ell' \in \NN$ and $L \in \llbracket |\ell - \ell'|, \ell + \ell' \rrbracket$ we have:
\[
\sum_{L = |\ell-\ell'|}^{\ell+\ell'} \Tilde{Q}^{(\ell, \ell'), L} \Tilde{Q}^{L, (\ell, \ell')} = I_{(2\ell+1)(2\ell'+1)}, \ \ \Tilde{Q}^{L, (\ell, \ell')} \Tilde{Q}^{(\ell, \ell'),L'} = \delta_{L,L'} I_{2L+1}
\]
where $\Tilde{Q}$ is either the complex ($\hat{Q}$) or real ($Q$) Clebsch Gordan tensor. 
\end{Theorem}

\begin{Theorem}[Symmetries of the complex Clebsch Gordan coefficients]
\label{th:clebsch_gordan_symmetries}
The complex Clebsch Gordan coefficients satisfy the following symmetry relations:
\[
\begin{aligned}
\hat{Q}^{(\ell,\ell'),L}_{m,m',M}
&
=
(-1)^{\ell+\ell'-L}
\hat{Q}^{(\ell,\ell'),L}_{-m,-m',-M}
\\
&
=
(-1)^{\ell+\ell'-L}
\hat{Q}^{(\ell',\ell),L}_{m',m,M}
\\
&
=
(-1)^{\ell-m}\sqrt{\frac{2L+1}{2\ell'+1}}\hat{Q}^{(\ell,L),\ell'}_{m,-M,-m'}
\\
&
=
(-1)^{\ell'+m'}\sqrt{\frac{2L+1}{2\ell+1}}\hat{Q}^{(L,\ell'),\ell}_{-M, m', -m}
\\
&
=
(-1)^{\ell-m}\sqrt{\frac{2L+1}{2\ell'+1}}\hat{Q}^{(L,\ell),\ell'}_{M, -m, m'}
\\
&
=
(-1)^{\ell'+m'}\sqrt{\frac{2L+1}{2\ell+1}}\hat{Q}^{(\ell',L),\ell}_{-m', M, m}.
\end{aligned}
\]
\end{Theorem}

\begin{Theorem}[Clebsch Gordan decomposition]
\label{th:clebsh_gordan_decomposition}
For any $L, \ell, \ell' \in \NN$ with $|\ell - \ell'| \leqslant L \leqslant \ell + \ell'$ we have the following decomposition of tensor products of Wigner matrices:
\[
\Tilde{D}^L \Tilde{Q}^{L,(\ell,\ell')} = \Tilde{Q}^{L,(\ell,\ell')} \Tilde{D}^{\ell} \otimes \Tilde{D}^{\ell'}
\]
where $\Tilde{D}$ is the complex (resp real) Wigner matrix and $\Tilde{Q}^{L,(\ell,\ell')}$ is the corresponding complex (resp. real) Clebsch-Gordan tensor.
\end{Theorem}
\begin{proof}
The result is often stated for complex Wigner matrices (see \cite{chirikjian2000engineering} for mode details), we show that the real case easily follows assuming the complex case:
\[
\begin{aligned}
Q^{L,(\ell,\ell')} D^{\ell} \otimes D^{\ell'}
&
=
\langle(C^{L} \hat{Q}^{L,(\ell,\ell')} (C^{\ell})^* \otimes (C^{\ell'})^*\rangle)
\langle
C^{\ell} \otimes C^{\ell'}
\hat{D}^{\ell} \otimes \hat{D}^{\ell'} (C^{\ell})^* \otimes (C^{\ell'})^*\rangle)
\\
&
=
C^{L} \hat{Q}^{L,(\ell,\ell')} \hat{D}^{\ell} \otimes \hat{D}^{\ell'} (C^{\ell})^* \otimes (C^{\ell'})^*
\\
&
=
C^{L} \hat{D}^L \hat{Q}^{L,(\ell,\ell')}  (C^{\ell})^* \otimes (C^{\ell'})^*
\\
&
=
\langle(C^{L} \hat{D}^L (C^{L})^{*}\rangle) \langle(C^{L} \hat{Q}^{L,(\ell,\ell')}  (C^{\ell})^* \otimes (C^{\ell'})^*\rangle)
\\
&
=
D^L Q^{L,(\ell,\ell')}
\end{aligned}
\]
\end{proof}

\subsubsection{Steerable bases}

\label{sec:steerable_bases}

\begin{Definition}[3D steerable kernel basis]
A $3D$ steerable kernel basis $\kappa$ is a finite dimensional basis of functions $\kappa_{rm}^{\ell}: \RR^3 \rightarrow \RR$ of the form:
\[
\kappa_{rm}^{\ell}(x) := \varphi_{r}^{\ell}(\Vert x \Vert_2^2) Y^{\ell}_m(x), \ \ \forall x \in \RR^3.
\]
where $\varphi_{r}^{\ell}: \RR \rightarrow \RR$ is the radial component, and the angular component is given by spherical harmonics. 
\end{Definition}

A typical choice for the radial component $\varphi_r$ (adopted by \cite{thomas2018tensor, weiler20183d}) is normalized Gaussian shell functions of the form:
\[
\varphi_r(t) := \frac{1}{t^{\ell / 2}} \exp \left( \frac{(\sqrt{t} - \rho_r^{\ell})^2}{2(\sigma_r^{\ell})^2}\right)
\]
where $\rho_r^{\ell}$ is the radius of the shall and the standard deviation parameter $\sigma_r^{\ell}$ controls its width. The $\frac{1}{t^{\ell / 2}}$ factor is normalizing the spherical harmonic component, as spherical harmonics are homogeneous polynomials we can normalize the input to the spherical harmonics instead. In this work we consider steerable kernels given by the Zernike basis \cite{lakshminarayanan2011zernike} (definition below) which also forms a Hilbert basis of the Hilbert space $L^{2}(\mathcal{B}_3)$ of square integrable functions on the unit ball $\mathcal{B}_3 \subset \RR^3$.

\begin{Definition}[Zernike polynomials (radial)]
For all $\ell, n \in \NN$ such that $n - \ell \geqslant 0$ and $2 | (n - \ell)$ the (radial) Zernike polynomial $R_{n}^{\ell}$ is a degree $k = (n - \ell) / 2$ polynomial function on the real line defined for all $x \in \RR$ by:
\[
R_{n}^{\ell}(x) = \sum_{v=0}^k R_{n,v}^{\ell} x^v
\]
whose coefficient $R_{n,v}^{\ell}$ is given by:
\[
R_{n,v}^{\ell} :=
\frac{(-1)^k}{2^{2k}} \sqrt{\frac{2\ell + 4k + 3}{3}} \begin{pmatrix}
2k \\
k
\end{pmatrix}
(-1)^v
\frac{
\begin{pmatrix}
k \\
v
\end{pmatrix}
\begin{pmatrix}
2(k+\ell+v)+1 \\
2k
\end{pmatrix}}{\begin{pmatrix}
k+\ell+v \\
k
\end{pmatrix}}.
\]
\end{Definition}

\begin{Theorem}[Orthogonality of (radial) Zernike polynomials]
For all $\ell \in \NN$ the type $\ell$ (radial) Zernike polynomials form an orthogonal family for the push forward of Lebesgue's volume measure on the unit ball to the unit interval $[0, 1[$ by the norm function. More precisely, for all $\ell \in \NN$ and all $n,n'\in \NN$ such that $n,n' \geqslant \ell$ and $2|(n-\ell)$, $2|(n'-\ell)$ we have:
\[
\int_{r=0}^1 R_n^{\ell}(r)R_{n'}^{\ell}(r) r^2 \mathrm{d}r = 
\delta_{n,n'}.
\]
\end{Theorem}

\begin{Definition}[3D Zernike polynomials]
For any $\ell,n,m$ such that $n - \ell \geqslant 0$, $2|(n-\ell)$ and $m \in \llbracket -\ell, \ell \rrbracket$ the 3D Zernike polynomial $Z^{\ell}_{nm}$ is a degree $n$ polynomial functions over $\RR^3$ defined for all $(x,y,z) \in \RR^3$ by:
\[
Z^{\ell}_{nm}(x,y,z) := R_n^{\ell}(x^2 + y^2 + z^2)Y_m^{\ell}(x,y,z).
\]
\end{Definition}

\begin{Theorem}[Hilbert basis of $L^2(\mathcal{B}_3)$]
The 3D Zernike polynomials, form a Hilbert basis of the Hilbert space $L^2(\mathcal{B}_3)$ of square integrable functions on the unit ball $\mathcal{B}_3 \subset \RR^3$, that is, for any $n,\ell,m$ and $n',\ell',m'$ such that the Zernike polynomials $Z_{nm}^{\ell}$ and $Z_{n'm'}^{\ell'}$ are defined we have:
\[
\frac{3}{4\pi}\int_{x \in \mathcal{B}_3} Z_{nm}^{\ell}(x)Z_{n'm'}^{\ell'}(x) \mathrm{d}x
=
\delta_{nn'} \delta_{\ell\ell'} \delta_{mm'}.
\]
\end{Theorem}

\section{Proofs of theorems}

\subsection{proof of \cref{th:gconv_equivariance}}
\label{sec:g_conv_eq_proof}
The proof is straightforward, the specificity of our definition is that the measure $\nu$ is not necessarily invariant, considering the action of $G$ on $\nu$ in our definition allows us to perform the change of variable $u = g^{-1}y$:  
\[
\begin{aligned}
(g.f) \ast_{g.\nu} \kappa 
:=
\int_G f(g^{-1}y)\kappa(y^{-1}x)d(g.\nu)(y)
=
\int_G f(u)\kappa(u^{-1}(g^{-1}x))d\nu(u)
=
g.(f \ast_{\nu} \kappa).
\end{aligned}
\]

\subsection{proof of \cref{th:gcnn_equivariance}}
\label{sec:gcnn_eq_proof}
We proceed by simple induction on the number of layers. We first prove the property for a single layer network, we have:
\[
\begin{aligned}
y^1(g.f,g.\nu_0)
&
= 
\xi(\mathrm{Conv}_{G}(g.f, g.\nu_0, W^1, b^1))
= 
\xi(W^1 (g.f) \ast_{g.\nu_0} \kappa + b^1)
\\
&
=
g.\xi(W^1f \ast_{\nu_0} \kappa + b^1)
=
g.y^1(f,\nu_0)
\end{aligned}
\]
we now assume that a any $n$ layer $G$-CNN is equivariant let $\nu_{<n} := (\nu_{0}, \dots, \nu_{n-1})$ the $n$-th layer statisfies $y^n(g.f,\nu_{<n}) = g.y^n(f,g.\nu_{<n})$ thus we have:
\[
\begin{aligned}
y^{n+1}(g.f,g.\nu)
&
= 
\xi(\mathrm{Conv}_{G}(y^n(g.f,g.\nu_{<n}), \nu_n, W^{n+1})) 
= 
\xi(W^{n+1} y^n(g.f,g.\nu_{<n}) \ast_{g.\nu_n} \kappa + b^{n+1})
\\
&
=
\xi(W^{n+1} (g.y^n(f,\nu_{<n})) \ast_{g.\nu_n} \kappa + b^{n+1})
=
g.\xi(W^{n+1} y^n(f,\nu_{<n}) \ast_{\nu_n} \kappa + b^{n+1})
\\
&
=
g.y^{n+1}(f,\nu).
\end{aligned}
\]
which concludes the proof of \cref{th:gcnn_equivariance} by induction.
\qed

\subsection{Equivariance of Dirac measure}
\label{sec:dirac_measure_eq_proof}
We prove that for a point cloud $X \subset \RR^3$ we have $R.\delta_{X} = \delta_{R.X}$ for any rotation $R \in \mathrm{SO}(3)$. For any measurable set $S \subseteq \RR^3$ we have:
\[
R.\delta_X(S) 
= 
\delta_X(R^{-1}.S)
=
\sum_{i} \delta_{X_i}(R^{-1}.S)
\]
we have $\delta_{X_i}(R^{-1}.S) = 1$ iff $X_i \in R^{-1}.S$ and $0$ otherwise, since $X_i \in R^{-1}.S$ iff $RX_i \in S$ we have $\delta_{X_i}(R^{-1}.S) = \delta_{R.X_i}(S)$ therefore:
\[
R.\delta_X(S) 
=
\sum_{i} \delta_{X_i}(R^{-1}.S)
=
\sum_{i} \delta_{RX_i}(S)
=
\delta_{R.X}(S)
\]
which proves the equality $R.\delta_{X} = \delta_{R.X}$. \qed

\subsection{proof of \cref{lemma:separable_conv}}
\[
\begin{aligned}
f \ast (\kappa_1 \otimes \kappa_2)(x, R)
&
=
\int_{\RR^3}   \int_{\mathrm{SO}(3)} f(t,H)\kappa_1(H^{-1}(x-t))\kappa_2(H^{-1}R) d\mu(H) d\lambda(t)
\\
&
=
   \int_{\mathrm{SO}(3)}
   \left(\int_{\RR^3} f(t,H)\kappa_1(H^{-1}(x-t))  d\lambda(t) \right)\kappa_2(H^{-1}R)d\mu(H)
\\
&
=
(f \ast (\kappa_1 \otimes \delta_I) ) \ast (\delta_0 \otimes \kappa_2).
\end{aligned}
\]
\qed

\subsection{proof of \cref{th:conv_approx}}
\[
\begin{aligned}
\mathrm{Conv}_{0\times \mathrm{SO}(3)}(\mathrm{Conv}_{\RR^3 \times I}(f, A), B)_i
&
=
\sum_{jm} B_{ijm} \left( \sum_{kl} A_{mkl} f_l \ast (\kappa_k \otimes \delta_I) \right) \ast (\delta_0 \otimes \theta_j)
\\
&
=
\sum_{jkl} \left(\sum_m B_{ijm} A_{mkl}\right) (f_l \ast (\kappa_k \otimes \delta_I)) \ast (\delta_0 \otimes \theta_j)
\\
&
=
\sum_{jkl} (B.A)_{ijkl} f_l \ast (\kappa_k \otimes \theta_j) 
=
\mathrm{Conv}_{\mathrm{SE}(3)}(f, B.A)_i
\end{aligned}
\]
\qed

\subsection{proof of \cref{th:multiview_cnn}}
We proceed by induction on the number of layers. We first prove the result for a single layer network:
\[
\begin{aligned}
\Tilde{y}^1_{\RR^3}(f, \lambda_0)(x,R) 
&
=
\xi(\mathrm{Conv}_{\RR^3}(R.f, R.\lambda_0, W^1, b^1)(x))
\\
&
=
\xi\left(W^1 \int_{\RR^3} f(R^{-1}t)\kappa(x - t) d(R.\lambda_0)(t) + b^1 \right)
\\
&
=
\xi\left(W^1 \int_{\RR^3} f(u)\kappa(R(R^{-1}x - u)) d\lambda_0(u) + b^1 \right)
\\
&
=
\xi\left(W^1 \int_{\RR^3} \Tilde{f}(u,R^{-1})\kappa(R(R^{-1}x - u)) d\lambda_0(u) + b^1 \right)
\\
&
=
\xi(
\mathrm{Conv}_{\RR^ \times I}(\Tilde{f},\lambda_0,W^1)(R^{-1}x, R^{-1}) + b^1)
\\
&
=
y^1_{\RR^3 \times I}(\Tilde{f},R.\lambda_0)(R^{-1}x, R^{-1})
\end{aligned}
\]
Now we assume the equality holds for networks with $n$ layers, that is $\Tilde{y}^{n}_{\RR^3}(f, \lambda_{<n})(x,R) = y^n_{\RR^3 \times I}(\Tilde{f}, \lambda_{<n})(R^{-1}x,R^{-1})$. We have:
\[
\begin{aligned}
\Tilde{y}^{n+1}(f, \lambda)(x,R)
&
=
\xi(\mathrm{Conv}_{\RR^3}(y^n_{\RR^3}(R.f, R.\lambda_{<n}), R.\lambda_n, W^{n+1}, b^{n+1})(x))
\\
&
=
\xi\left( 
W^{n+1} y^n_{\RR^3}(R.f, R.\lambda_{<n}) \ast_{\RR^3, R.\lambda_n} \kappa(x) + b^{n+1} 
\right)
\\
&
=
\xi\left( 
W^{n+1} 
\int_{\RR^3} y^n_{\RR^3}(R.f, R.\lambda_{<n})(t)\kappa(x-t)d(R.\lambda_n)(t)
+ b^{n+1} 
\right)
\\
&
=
\xi\left( 
W^{n+1} 
\int_{\RR^3} \Tilde{y}^n(f, \lambda_{<n})(t,R)\kappa(x-t)d(R.\lambda_n)(t)
+ b^{n+1} 
\right)
\\
&
=
\xi\left( 
W^{n+1} 
\int_{\RR^3} y_{\RR^3 \times I}^n(f, \lambda_{<n})(R^{-1}t,R^{-1})\kappa(x-t)d(R.\lambda_n)(t)
+ b^{n+1} 
\right)
\\
&
=
\xi\left( 
W^{n+1} 
\int_{\RR^3} y_{\RR^3 \times I}^n(f, \lambda_{<n})(u,R^{-1})\kappa(R(R^{-1}x-u))d\lambda_n(u)
+ b^{n+1} 
\right)
\\
&
=
\Tilde{y}^{n+1}(\Tilde{f}, \lambda)(R^{-1}x, R^{-1}).
\end{aligned}
\]
which concludes the proof of \cref{th:multiview_cnn} by induction.
\qed
\subsection{proof of \cref{lemma:harmonic_sep_se3_conv}}

We first compute the Wigner decomposition of the $\RR^3$ component, by uniqueness of the Wigner coefficients the matrices of coefficients $f \ast (\kappa^{\ell'}_{rm'} \otimes \delta_I)^L$ is uniquely determined by: 
\[
 f \ast (\kappa^{\ell'}_{rm'} \otimes \delta_I)(x, R) = 
\sum_{L \geqslant 0} \langle f \ast (\kappa^{\ell'}_{rm'} \otimes \delta_I)^L(x), D^L(R)\rangle
\]
Denoting by $e_{m'}^{\ell'} \in \RR^{2\ell'+1}$ the $m'$-th canonical vector, $e_{m',i}^{\ell'} = 1$ iff $i = m'$ and $0$ otherwise, we have:
\[
\begin{aligned}
f \ast (\kappa^{\ell'}_{rm'} \otimes \delta_I) (x, R)
&
=
\int_{\RR^3} f(t,R) \kappa^{\ell'}_{rm'}(R^{-1}(x-t))\mathrm{d}t
\\
&
=
\sum_{\ell=0}^{+\infty} \int_{\RR^3} \langle f^{\ell}(t), D^{\ell}(R) \rangle \sum_{k'=-\ell'}^{\ell'}D^{\ell'}(R)^{\top}_{m'k'} \kappa^{\ell'}_{rk'}(x-t)\mathrm{d}t
\\
&
=
\sum_{\ell=0}^{+\infty} \int_{\RR^3} \langle f^{\ell}(t), D^{\ell}(R) \rangle  \langle D^{\ell'}(R)^{\top} \kappa^{\ell'}_{r}(x-t), e_{m'}^{\ell'} \rangle \mathrm{d}t
\\
&
=
\sum_{\ell=0}^{+\infty} \int_{\RR^3} \langle f^{\ell}(t), D^{\ell}(R) \rangle  \langle D^{\ell'}(R)^{\top} , e_{m'}^{\ell'} \kappa^{\ell'}_{r}(x-t)^{\top} \rangle \mathrm{d}t
\\
&
=
\sum_{\ell=0}^{+\infty} \int_{\RR^3} \langle f^{\ell}(t), D^{\ell}(R) \rangle  \langle \kappa^{\ell'}_{r}(x-t)(e_{m'}^{\ell'})^{\top}, D^{\ell'}(R)  \rangle \mathrm{d}t
\\
&
=
\sum_{\ell=0}^{+\infty} \int_{\RR^3} \langle f^{\ell}(t) \otimes \kappa^{\ell'}_{r}(x-t)(e_{m'}^{\ell'})^{\top}, D^{\ell}(R) \otimes D^{\ell'}(R) \rangle \mathrm{d}t
\\
&
=
\sum_{\ell=0}^{+\infty}  \sum_{L=|\ell - \ell'|}^{\ell+\ell'}\int_{\RR^3} \langle f^{\ell}(t) \otimes \kappa^{\ell'}_{r}(x-t)(e_{m'}^{\ell'})^{\top}, Q^{(\ell, \ell'),L}D^L(R)Q^{L, (\ell, \ell')}  \rangle \mathrm{d}t
\\
&
=
\sum_{L=0}^{+\infty} \sum_{\ell}^{}\int_{\RR^3} \langle Q^{L, (\ell, \ell')} f^{\ell}(t) \otimes \kappa^{\ell'}_{r}(x-t)(e_{m'}^{\ell'})^{\top} Q^{ (\ell, \ell'), L}, D^L(R)  \rangle \mathrm{d}t
\\
&
=
\sum_{L=0}^{+\infty}  \langle \sum_{\ell}^{}  Q^{L, (\ell, \ell')} \int_{\RR^3} f^{\ell}(t) \otimes \kappa^{\ell'}_{r}(e_{m'}^{\ell'})^{\top}(x-t) \mathrm{d}t \  Q^{ (\ell, \ell'), L}, D^L(R)  \rangle 
\\
&
=
\sum_{L=0}^{+\infty}   \langle \sum_{\ell}^{}   Q^{L, (\ell, \ell')}  \left(f^{\ell} \ast_{\RR^3} \kappa^{\ell'}_{r}(x)\right) \left(I_{2\ell+1} \otimes (e_{m'}^{\ell'})^{\top}\right) Q^{ (\ell, \ell'), L}, D^L(R)  \rangle
\end{aligned}
\]
where the sums over $\ell$ are taken over the values of $\ell$ such that $|\ell - \ell'| \leqslant L \leqslant \ell + \ell'$ thus:
\[
\begin{aligned}
f \ast (\kappa^{\ell'}_{rm'} \otimes \delta_I)^L_{ij}
&
=
\left(\sum_{\ell}^{}   Q^{L, (\ell, \ell')}  \left(f^{\ell} \ast_{\RR^3} \kappa^{\ell'}_{r}(x)\right) \left(I_{2\ell+1} \otimes (e_{m'}^{\ell'})^{\top}\right) Q^{ (\ell, \ell'), L}\right)_{ij}
\\
&
=
\sum_{\ell}^{}   Q^{L, (\ell, \ell')}_{i,:,:}  \left(f^{\ell} \ast_{\RR^3} \kappa^{\ell'}_{r}(x)\right) \left(\left(I_{2\ell+1} \otimes (e_{m'}^{\ell'})^{\top}\right) Q^{ (\ell, \ell'), L}\right)_{:,j}
\\
&
=
\sum_{\ell}^{}   Q^{L, (\ell, \ell')}_{i,:,:}  \left(f^{\ell} \ast_{\RR^3} \kappa^{\ell'}_{r}(x)\right) Q^{ (\ell, \ell'), L}_{:,m',j}
\\
&
=
\sum_{\ell,m}^{}   Q^{L, (\ell, \ell')}_{i,:,:}  \left(f^{\ell}_{:,m} \ast_{\RR^3} \kappa^{\ell'}_{r}(x)\right) Q^{ (\ell, \ell'), L}_{m,m',j}
\end{aligned}
\]

We now compute the Wigner decomposition of the $\mathrm{SO}(3)$ component. First observe that the $\mathrm{SO}(3)$ component $f \ast (\delta_0 \otimes \theta)(t,R)$ consist of an $\mathrm{SO}(3)$ convolution between $f(t,\bullet)$ and $\theta$:
\[
\begin{aligned}
f \ast (\delta_0 \otimes \theta)(x,R)
&
:= 
\int_{\mathrm{SO}(3)} \hspace{-3mm} f(x,H)\theta(H^{-1}R) d\mu(H)
\\
&
=
f(x,\bullet) \ast_{\mathrm{SO}(3)} \theta (R) 
\end{aligned}
\]
thus we have:
\begin{equation}
\label{eq:so3_comp_observation}
f \ast (\delta_0 \otimes \theta)^L(x)
=
\left(f(x,\bullet) \ast_{\mathrm{SO}(3)} \theta\right)^L
\end{equation}
it suffice to compute the Wigner decomposition decomposition of $\mathrm{SO}(3)$ convolution. This has been done in \cite{cohen2018spherical} but we also provide the proof since it is relatively straightforward:
\begin{Theorem}
\label{th:so3_conv_wigner_coeffs}
The the Wigner coefficients of the convolution between two smooth functions $f,g: \mathrm{SO}(3) \rightarrow \RR$ are given by matrix multiplication:
$
(f \ast_{\mathrm{SO}(3)} g)^{\ell} = f^{\ell} g^{\ell}.
$
\end{Theorem}
\begin{proof}
We again rely on the uniqueness of the Wigner coefficients:
\[
\begin{aligned}
f \ast_{\mathrm{SO}(3)} g(R) 
&
= 
\int \sum_{\ell} \langle f^{\ell}, D^{\ell}(H) \rangle \sum_{\ell'} \langle g^{\ell'}, D^{\ell'}(H^{-1}R) \rangle dH
\\
&
=
\int \sum_{\ell} \langle f^{\ell}, D^{\ell}(H) \rangle \sum_{\ell'} \langle g^{\ell'}, D^{\ell'}(H)^{\top}D^{\ell'}(R) \rangle dH
\\
&
=
\int \sum_{\ell} \langle f^{\ell}, D^{\ell}(H) \rangle \sum_{\ell'} \langle D^{\ell'}(R)(g^{\ell'})^{\top},  D^{\ell'}(H) \rangle dH
\\
&
=
 \sum_{\ell\ell'} \int \langle f^{\ell} \otimes D^{\ell'}(R)(g^{\ell'})^{\top}, D^{\ell}(H) \otimes D^{\ell'}(H) \rangle \rangle dH
 \\
 &
 =
 \sum_{\ell} \int \langle f^{\ell} \otimes D^{\ell}(R)(g^{\ell})^{\top}, D^{\ell}(H) \otimes D^{\ell}(H) \rangle \rangle dH
 =
 \sum_{\ell} \sum_{ij} f^{\ell}_{ij} (D^{\ell}(R) (g^{\ell})^{\top})_{ij}
\\
&
=
\sum_{\ell} \langle f^{\ell}, D^{\ell}(R) (g^{\ell})^{\top} \rangle
=
\sum_{\ell} \langle f^{\ell}g^{\ell} , D^{\ell}(R) \rangle
\end{aligned}
\]
\end{proof}
Putting \cref{eq:so3_comp_observation} and \cref{th:so3_conv_wigner_coeffs} together we obtain:
\[
(f \ast (\delta_0 \otimes \theta))^{L}(x) =  f^{L}(x)\theta^{L}
\]
for all $x \in \RR^3$ which concludes the proof of \cref{lemma:harmonic_sep_se3_conv}
\qed

\subsection{proof of \cref{th:tfn_vs_se3_conv_1}}
We only need to prove the Wigner decomposition of $\mathrm{Conv}_{\mathrm{SE}(3)}(f,\lambda \otimes \nu, W, b)$ as we only added the TFN expression for side by side comparison. We assume a kernels basis of the form $\kappa^{\ell}_{rm'} \otimes D^{L}_{Mn}$ where $\kappa$ is a steerable kernel basis. By the general expression of $G$-conv layers from \cref{def:g_conv} and its specialization to $\mathrm{SE}(3)$ we have:
\begin{equation}
\begin{aligned}
\mathrm{Conv}_{\mathrm{SE}(3)}(f, \lambda \otimes \mu, W, b)_d
&
=
\sum_{\ell'Lcrm'Mn} 
f_{c}
\ast_{\mathrm{SE}(3),\lambda \otimes \mu} \kappa^{\ell'}_{rm'} \otimes D^{L}_{Mn}
W^{\ell', L}_{nd,crm'M} + b_d
\end{aligned}
\end{equation}
we can decompose the $d$-th output function 
$\mathrm{Conv}_{\mathrm{SE}(3)}(f, \lambda \otimes \mu, W, b)_d: \mathrm{SE}(3) \rightarrow \RR$
in the Wigner basis:
\begin{equation}
\label{eq:wigner_decomposition_eq_1}
\begin{aligned}
\mathrm{Conv}_{\mathrm{SE}(3)}(f, \lambda \otimes \mu, W, b)^J_{ijd}
=
\sum_{crm' \ell'L Mn} 
\left(
f_{c}
\ast_{\mathrm{SE}(3),\lambda \otimes \mu} \kappa^{\ell'}_{rm'} \otimes D^{L}_{Mn}
\right)^J_{ijd}W^{\ell', L}_{nd,crm'M} + b^J_d
\end{aligned}
\end{equation}
since the bias term $b_d$ is constant (we view it as a constant function over $\mathrm{SE}(3)$) its Wigner decomposition is null except for type $0$ Wigner matrices, that is, $b^J_d = b_d$ iff $J = 0$ and $0$ otherwise. We can now focus on the main term, we use separability of $\mathrm{SE}(3)$ convolution (\cref{lemma:separable_conv}) to decompose it into a $\RR^3$ and $\mathrm{SO}(3)$ convolution: 
\begin{equation}
\begin{aligned}
f_{c}
\ast_{\mathrm{SE}(3),\lambda \otimes \mu} \kappa^{\ell}_{rm'} \otimes D^{L}_{Mn}
=
(f_{c} \ast_{\lambda} \kappa^{\ell}_{rm'} \otimes \delta_I) \ast_{\mu} \delta_0 \otimes D^{L}_{Mn} 
\end{aligned}
\end{equation}
We proceed using \cref{lemma:harmonic_sep_se3_conv} to get the Wigner decomposition starting with the outer $\mathrm{SO}(3)$-convolution by $D^L_{Mn}$ (first line). We observe that this only has non zero coefficients w.r.t. $D^L$ thus the expression can non zero only when $J = L$. Then we decompose the $\RR^3$ component (second line), we obtain the final expression on the (last line) using the equality:
\[
\left(Q^{(\ell,\ell'),J}_{m,m',:} e^J_{Mn}\right)_{j} 
= 
\sum_{a} Q^{(\ell,\ell'),J}_{m,m',a} (e^J_{Mn})_{aj}
=
\sum_a Q^{(\ell,\ell'),J}_{m,m',a} \delta_{Ma}\delta_{nj}
=
Q^{(\ell,\ell'),J}_{m,m',M}\delta_{nj}
\]
we obtain: 
\begin{equation}
\label{eq:wigner_decomposition_eq_2}
\begin{aligned}
((f_{c} \ast_{\lambda} \kappa^{\ell}_{rm} \otimes \delta_I) \ast_{\mu} \delta_0 \otimes D^{L}_{Mn})^J_{ij}
&
=
 \left(\delta_{JL}(f_{c} \ast_{\lambda} \kappa^{\ell}_{rm} \otimes \delta_I)^J e_{Mn}^J\right)_{ij}
\\
&
=
\delta_{JL} 
\sum_{\ell m} Q^{J,(\ell,\ell')}_{i,:,:} f_{:,m,c}^{\ell} \ast_{\lambda} \kappa^{\ell'}_{r} \left(Q^{(\ell,\ell'),J}_{m,m',:}
e_{Mn}^J\right)_j
\\
&
=
\delta_{JL} 
\sum_{\ell m} Q^{J,(\ell,\ell')}_{i,:,:} f_{:,m,c}^{\ell} \ast_{\lambda} \kappa^{\ell'}_{r} Q^{(\ell,\ell'),J}_{m,m',M}\delta_{nj}
\end{aligned}
\end{equation}
where $e^J_{Mj}$ is the $2J+1$ by $2J+1$ defined by $(e^J_{Mj})_{ab} := \delta_{Ma}\delta_{jb}$.

By plugging the result of \cref{eq:wigner_decomposition_eq_2} into the expression of the $\mathrm{Conv}_{\mathrm{SE}(3)}$ layer Wigner coefficients in \cref{eq:wigner_decomposition_eq_1} we obtain the desired expression:
\[
\begin{aligned}
\mathrm{Conv}_{\mathrm{SE}(3)}(f, \lambda \otimes \mu, W, b)^L_{ijd}
&
=
\sum_{\ell'crm'Mn} 
\left(\sum_{\ell m} Q^{L,(\ell,\ell')}_{i,:,:} f_{:,m,c}^{\ell} \ast_{\lambda} \kappa^{\ell'}_{r} Q^{(\ell,\ell'),L}_{m,m',M}\delta_{nj} \right) W^{\ell', L}_{nd,crm'M}
+ b^L_d
\\
&
=
\sum_{\ell\ell' crm} Q^{L,(\ell,\ell')}_{i,:,:} \left( f^{\ell}_{:,m,c} \ast_{\lambda} \kappa_r^{\ell'} \right) \sum_{m'M} Q^{(\ell,\ell'),L}_{m,m',M}\sum_n \delta_{nj} W^{\ell',L}_{nd,crm'M} + b^L_d
\\
&
=
\sum_{\ell\ell' crm} Q^{L,(\ell,\ell')}_{i,:,:} \left( f^{\ell}_{:,m,c} \ast_{\lambda} \kappa_r^{\ell'} \right) \sum_{m'M} Q^{(\ell,\ell'),L}_{m,m',M} W^{\ell',L}_{jd,crm'M} + b^L_d.
\end{aligned}
\]
which concludes the proof of \cref{th:tfn_vs_se3_conv_1}.
\qed

\subsection{proof of \cref{th:tfn_vs_se3_conv_2}}

The first equality $\mathrm{Conv}_{\mathrm{SE}(3)}(f, \nu, W, b)
 = 
\mathrm{TFN}(f, \nu, \iota(W), b)$ directly follows from \cref{th:tfn_vs_se3_conv_1}. We now prove that the maps $\iota$ and $\iota^{-1}$ are inverse to each other, the second equality would follow immediately from the first one under this assumption as:
\[
\begin{aligned}
\mathrm{TFN}(f, \nu, V, b)
=
\mathrm{TFN}(f, \nu, \iota(\iota^{-1}(V)), b)
=
\mathrm{Conv}_{\mathrm{SE}(3)}(f,\nu, \iota^{-1}(V), b).
\end{aligned}
\]
We will need the following lemma which combines the orthogonality and symmetry relations of Clebsch Gordan coefficients from \cref{th:clebsch_gordan_orthogonality} and \cref{th:clebsch_gordan_symmetries} respectively:
\begin{Lemma}[Clebsch Gordan orthogonality (bis)]
\label{lemma:clebsch_gordan_orthogonality}
The real and complex ($\Tilde{Q} = \hat{Q}$ or $\Tilde{Q} = Q$) satisfy the following orthogonality relations:
\[
\sum_{\ell,m} \Tilde{Q}^{L,(\ell,\ell')}_{N,m,n'}\Tilde{Q}^{(\ell,\ell'),L}_{m,m',M} 
= 
\frac{2L+1}{2\ell+1} \delta_{m'n'}\delta_{MN}
\]
\end{Lemma}
\begin{proof}
First we observe that orthogonality of the real Clebsch Gordan tensors follow from the complex case as the complex to real transition matrices $C$ simplify. For simplicity we transpose the Clebsch Gordan tensors by setting $\Tilde{T}^{(\ell',L),\ell}_{m',M,m} := \Tilde{Q}^{(\ell, \ell'),L}_{m,m',M}$ and $\Tilde{T}^{\ell, (\ell',L)}_{m,m',M} = \Tilde{T}^{(\ell',L), \ell}_{m',M,m}$ where $\Tilde{Q} = \hat{Q}$ or $\Tilde{Q} = Q$ is either the complex or real Clebsch Gordan tensor. Assuming the orthogonality relations hold in the complex case and using the relation between real and complex Clebsch Gordan tensors from \cref{def:clebsch_gordan_coefficients} we have:
\[
\begin{aligned}
\sum_{\ell,m} Q^{L,(\ell,\ell')}_{N,m,n'}Q^{(\ell,\ell'),L}_{m,n',N}
&
=
\sum_{\ell,m} T^{(\ell',L)\ell}_{m',M,m} T^{\ell, (\ell',L)}_{m,n',N}
=
\sum_{\ell} \left(T^{(\ell',L)\ell} T^{\ell, (\ell',L)}\right)_{m'n'MN}
\\
&
=
\sum_{\ell} \left((C^{\ell'})^* \otimes C^L \hat{T}^{(\ell',L)\ell} (C^{\ell})^* C^{\ell} \hat{T}^{\ell, (\ell',L)} C^{\ell'} \otimes (C^L)^*\right)_{m'n'MN}
\\
&
=
 \left((C^{\ell'})^* \otimes C^L \left(\sum_{\ell} \hat{T}^{(\ell',L)\ell}  \hat{T}^{\ell, (\ell',L)}\right) C^{\ell'} \otimes (C^L)^*\right)_{m'n'MN}
\\
&
=
\frac{2L+1}{2\ell+1}\left((C^{\ell'})^* \otimes C^L C^{\ell'} \otimes (C^L)^*\right)_{m'n'MN}
=
\frac{2L+1}{2\ell+1}I_{m'n'MN}
\\
&
=
\frac{2L+1}{2\ell+1} \delta_{m'n'}\delta_{MN}
\end{aligned}
\]
\[
\begin{aligned}
\sum_{m',M} Q^{L,(\ell,\ell')}_{M,n,m'}Q^{(\ell,\ell'),L}_{m,m',M}
&
=
\sum_{m',M}  T^{\ell, (\ell',L)}_{m,m',M} T^{(\ell',L)\ell}_{n',N,n}
=
\left( T^{\ell, (\ell',L)} T^{(\ell',L)\ell} \right)_{mn}
\\
&
=
\left( 
C^{\ell} T^{\ell, (\ell',L)} C^{\ell'} \otimes (C^L)^*  (C^{\ell'})^* \otimes C^L T^{(\ell',L)\ell} (C^{\ell})^*
\right)_{mn}
\\
&
=
\left( 
C^{\ell} T^{\ell, (\ell',L)} T^{(\ell',L)\ell} (C^{\ell})^*
\right)_{mn}
=
\frac{2L+1}{2\ell+1}I_{mn}
\\
&
=
\frac{2L+1}{2\ell+1}\delta_{mn}
\end{aligned}
\]
We now prove the orthogonality relations for the complex Clebsch Gordan tensors. By the symmetry relations of \cref{th:clebsch_gordan_symmetries} we have 
\[
\hat{Q}^{(\ell, \ell'),L}_{m,m',M}
=
(-1)^{\ell'+m'}\sqrt{\frac{2L+1}{2\ell+1}}\hat{Q}^{(\ell',L),\ell}_{-m', M, m}
\]
combining the above with the orthogonality relations of \cref{th:clebsch_gordan_orthogonality} we obtain:
\[
\begin{aligned}
\sum_{\ell,m} \hat{Q}^{(\ell,\ell'),L}_{m,m',M}\hat{Q}^{L,(\ell,\ell')}_{N,m,n'}
&
=
\sum_{\ell,m}\frac{2L+1}{2\ell+1}
\hat{Q}^{(\ell',L),\ell}_{-m', M, m}
\hat{Q}^{\ell,(\ell',L)}_{m, -n', N}
\\
&
=
\frac{2L+1}{2\ell+1} \delta_{m'n'}\delta_{MN}
\end{aligned}
\]
similarly we have:
\[
\begin{aligned}
\sum_{m',M} \hat{Q}^{L,(\ell,\ell')}_{M,n,m'}\hat{Q}^{(\ell,\ell'),L}_{m,m',M}
&
=
\sum_{m'M}\frac{2L+1}{2\ell+1}
\hat{Q}^{\ell,(\ell',L)}_{n, -m', M}
\hat{Q}^{(\ell',L),\ell}_{-m', M, m}
\\
&
=
\frac{2L+1}{2\ell+1} \delta_{mn}
\end{aligned}
\]
which concludes the proof.
\end{proof}
We can now prove $\iota^{-1}(\iota(W)) = W$ and $\iota(\iota^{-1}(V)) = V$, using \cref{lemma:clebsch_gordan_orthogonality} we have:
\[
\begin{aligned}
\iota^{-1}(\iota(W))^{\ell',L}_{jd,crm'M}
&
=
\frac{2\ell+1}{2L+1}\sum_{\ell,m} Q^{L,(\ell,\ell')}_{M,m,m'} \iota(W)^{(\ell,\ell'),L}_{jd,crm} \\
&
=
\frac{2\ell+1}{2L+1}\sum_{\ell,m} Q^{L,(\ell,\ell')}_{M,m,m'} 
\sum_{n'N} Q^{(\ell, \ell'),L}_{m,n',N} W^{\ell', L}_{jd,crn'N}
\\
&
=
\sum_{n'N} \frac{2\ell+1}{2L+1}\left(\sum_{\ell,m} Q^{L,(\ell,\ell')}_{M,m,m'} 
 Q^{(\ell, \ell'),L}_{m,n',N}\right) W^{\ell', L}_{jd,crn'N}
\\
&
=
\sum_{n'N} \frac{2\ell+1}{2L+1} \frac{2L+1}{2\ell+1} \delta_{m'n'}\delta_{MN} W^{\ell', L}_{jd,crn'N}
\\
&
=
W^{\ell', L}_{crm'M,jd}
\end{aligned}
\]
\[
\begin{aligned}
\iota(\iota^{-1}(V))^{(\ell,\ell'),L}_{jd,crm}
&
=
\sum_{m'M} Q_{m,m',M}^{ (\ell, \ell'), L} \iota^{-1}(V)^{\ell', L}_{jd,crm'M}
=
\sum_{m'M} Q_{m,m',M}^{ (\ell, \ell'), L}\sum_{\ell,n}\frac{2\ell+1}{2L+1}Q^{L,(\ell,\ell')}_{n,m',M} V^{(\ell, \ell'),L}_{jd,crn}
\\
&
=
\sum_{\ell,n}
\frac{2\ell+1}{2L+1} \left(\sum_{m'M} Q_{m,m',M}^{ (\ell, \ell'), L}
Q^{L,(\ell,\ell')}_{n,m',M}\right) V^{(\ell, \ell'),L}_{jd,crn}
\\
&
=
\sum_{\ell,n}
\frac{2\ell+1}{2L+1} \frac{2L+1}{2\ell+1} \delta_{mn}
 V^{(\ell, \ell'),L}_{jd,crn}
\\
&
=
V^{(\ell, \ell'),L}_{jd,crm}
\end{aligned}
\]
which concludes the proof of \cref{th:tfn_vs_se3_conv_2}.
\qed

\subsection{proof of \cref{th:activations_equivariance}}
We prove $\mathrm{SE}(3)$ equivariance of the layers:
\[
\begin{aligned}
&
\mathcal{F} \circ \xi \circ \mathcal{F}^{-1}( \mathbf{CNN}_{\mathrm{SE}(3)}((t,R).f, (t,R).\nu, W, b))^L(x)
\\
&
=
D(R)^L\mathcal{F} \circ \xi \circ \mathcal{F}^{-1}(\mathbf{CNN}_{\mathrm{SE}(3)}(f,\nu, W, b))^L(R^{-1}x - t)
\end{aligned}
\]
the result for networks will follow by induction as in the proof of \cref{th:gcnn_equivariance}. For $\mathcal{F} = \mathcal{W}$ being the "continuous" Wigner transform we have:
\[
\mathcal{W}^{-1}(\mathbf{CNN}_{\mathrm{SE}(3)}(f, \nu, W, b)) 
=
\mathrm{CNN}_{\mathrm{SE}(3)}(f, \nu, W, b)
\]
using the above and specializing \cref{th:gcnn_equivariance} to $\mathrm{SE}(3)$ we have:
\[
\begin{aligned}
&
\mathcal{W} \circ \xi(
\mathcal{W}^{-1}(\mathbf{CNN}_{\mathrm{SE}(3)}((t,R).f, (t,R).\nu, W, b))
)^{L}(x)
\\
&
=
\mathcal{W} \circ \xi(
\mathrm{CNN}_{\mathrm{SE}(3)}((t,R).f, (t,R).\nu, W, b)
)^{L}(x)
\\
&
=
\mathcal{W} \circ \xi(
(t,R).\mathrm{CNN}_{\mathrm{SE}(3)}(f, \nu, W, b)
)^{L}(x)
\\
&
=
\mathcal{W}((t,R).\xi(
\mathrm{CNN}_{\mathrm{SE}(3)}(f, \nu, W, b)
))^{L}(x)
\\
&
=
\int_{\mathrm{SO}(3)} 
\hspace{-5mm}
(t,R).\xi(
\mathrm{CNN}_{\mathrm{SE}(3)}(f, \nu, W, b)
)(x,H) D^L(H) d\mu(H) 
\\
&
=
\int_{\mathrm{SO}(3)} 
\hspace{-5mm} \xi(
\mathrm{CNN}_{\mathrm{SE}(3)}(f, \nu, W, b)
)(R^{-1}x - t,R^{-1}H) D^L(H) d\mu(H)
\\
&
=
\int_{\mathrm{SO}(3)} 
\hspace{-5mm}
\xi(
\mathrm{CNN}_{\mathrm{SE}(3)}(f, \nu, W, b)
)(R^{-1}x - t,M) D^L(RM) d\mu(M)
\\
&
=
D^L(R)\int_{\mathrm{SO}(3)} 
\hspace{-5mm}
\xi(
\mathrm{CNN}_{\mathrm{SE}(3)}(f, \nu, W, b)
)(R^{-1}x - t,M) D^L(M) d\mu(M)
\\
&
=
D^L(R)\mathcal{W} \circ \xi(
\mathcal{W}^{-1}(\mathbf{CNN}_{\mathrm{SE}(3)}(f, \nu, W, b))
)^{L}(R^{-1}x - t)
\end{aligned}
\]
which proves equivalence for $\mathcal{F} = \mathcal{W}$. We now turn to the discrete case, 
let $G = \{R_1, \dots, R_n\} \subset \mathrm{SO}(3)$ a finite subgroup of $\mathrm{SO}(3)$ we denote by $\mathcal{W}_{G}$ the discrete Wigner transform (replacing $\mu$ by $\frac{1}{|G|}\sum_i \delta_{R_i}$. We have:
\[
\begin{aligned}
&
\mathcal{W}_G \circ \xi(
\mathcal{W}^{-1}(\mathbf{CNN}_{\mathrm{SE}(3)}((t,R).f, (t,R).\nu, W, b))
)^{L}(x)
\\
&
=
\mathcal{W}_G((t,R).\xi(
\mathrm{CNN}_{\mathrm{SE}(3)}(f, \nu, W, b)
))^{L}(x)
\\
&
=
\frac{1}{|G|} \sum_i 
(t,R).\xi(
\mathrm{CNN}_{\mathrm{SE}(3)}(f, \nu, W, b)
)(x,R_i) D^L(R_i)
\\
&
=
\frac{1}{|G|} \sum_i 
\xi(
\mathrm{CNN}_{\mathrm{SE}(3)}(f, \nu, W, b)
)(R^{-1}x-t,R^{-1}R_i) D^L(R_i)
\\
&
=
\frac{1}{|G|} \sum_i 
\xi(
\mathrm{CNN}_{\mathrm{SE}(3)}(f, \nu, W, b)
)(R^{-1}x-t,R_{\sigma{i}}) D^L(R_i)
\\
&
=
\frac{1}{|G|} \sum_j 
\xi(
\mathrm{CNN}_{\mathrm{SE}(3)}(f, \nu, W, b)
)(R^{-1}x,R_{j}) D^L(R_{\sigma^{-1}(j)})
\\
&
=
\frac{1}{|G|} \sum_j 
\xi(
\mathrm{CNN}_{\mathrm{SE}(3)}(f, \nu, W, b)
)(R^{-1}x,R_{j}) D^L(RR_{j})
\\
&
=
D^L(R) 
\frac{1}{|G|} \sum_j 
\xi(
\mathrm{CNN}_{\mathrm{SE}(3)}(f, \nu, W, b)
)(R^{-1}x,R_{j}) D^L(R_{j})
\\
&
=
D^L(R)\mathcal{W} \circ \xi(
\mathcal{W}^{-1}(\mathbf{CNN}_{\mathrm{SE}(3)}(f, \nu, W, b))
)^{L}(R^{-1}x-t)
\end{aligned}
\]
which concludes the proof of \cref{th:activations_equivariance}.
\qed

\end{document}